\documentclass[12pt,journal,onecolumn,a4paper]{IEEEtran}
\pdfoutput=1 
\usepackage{graphicx}  

\usepackage{amssymb}
\usepackage{amsmath}
\usepackage{amsthm}
\usepackage{hyperref}

\newtheorem{lemma}{Lemma}[section]

\newtheorem{definition}{Definition}[section]
\newtheorem{theorem}{Theorem}[section]
 
\begin{document}

  \title{Detection and Amelioration of Social Engineering Vulnerability in Contingency Table Data using an Orthogonalised Log-linear Analysis} 
 \author{Glynn Rogers,\thanks{Glynn Rogers,~Independent~Researcher,~glynn.rogers@internode.on.net,~Corresponding~author,}~ Malcolm Crompton\thanks{Malcolm~Crompton,~Information Integrity Solutions Pty Ltd, mcrompton@iispartners.com},~Gaurav Sapre\thanks{Gaurav~Sapre,~Data61,~CSIRO,~Australia,~Gaurav.Sapre@hotmail.com},~and~Jonathan Chan\thanks{Jonathan~Chan,~Data61,~CSIRO,~Australia,~Jonathan.Chan@data61.csiro.au}}
 \maketitle
 
\begin{abstract}
Social Engineering has emerged as a significant threat in cyber security. In a dialog based attack,  by having enough of a potential victim's personal data to be convincing, a social engineer impersonates the victim in order to manipulate the attack's target into revealing sufficient information for accessing the victim's accounts etc. We utilise the developing understanding of human information processing in the Information Sciences to characterise the vulnerability of the target to manipulation and to propose a form of countermeasure. Our focus is on the possibility of the social engineer being able to build the victim's profile by, in part, inferring personal attribute values from statistical information available either informally, from general knowledge, or, more formally, from some public database. We use an orthogonalised log linear analysis of data in the form of a contingence table to develop a measure of how susceptible particular subtables are to probabilistic inference as the basis for our proposed countermeasure.  This is based on the observation that inference relies on a high degree of non-uniformity and exploits the orthogonality of the analysis to define the measure in terms of subspace projections. 
\end{abstract}
\begin{IEEEkeywords}
cyber security, social engineering, data privacy, identity, personal attributes, identity impersonation, contingency table, log linear, geometric marginalisation, de-personalisation.
%\MSC[2020] 15A03 62H17 68T09
\end{IEEEkeywords}

\section{Introduction}

\subsection{The Social Engineering Problem}

Social engineering attacks on cyber infrastructure constitute a large and growing economic problem. As the comprehensive literature review in \cite{Wang} reveals, the term `social engineering' has been applied to a wide variety of illegitimate methods of gaining advantage ranging from simple breaches of physical security to sophisticated phishing campaigns \cite{phishing_attacks}. However the central feature of social engineering that emerges is the exploitation of human vulnerability through some form of interaction between attacker and the attack target. 

This is the case even where an unauthorised person gains physical entry to secure facilities by following legitimate personnel through a controlled access point (`tailgating') or enlisting the aid of an authorised person by pretending to be struggling with a heavy load \cite{Ghafir} \cite{Wang3}. Here, the attacker exploits the reluctance of people to confront a potentially awkward situation or the desire to be helpful which illustrates the general observation that humans are the weakest link in any security system \cite{Ghafir}.

Phishing is also becoming a major threat to cyber security. From the discussion in \cite{phishing_attacks} it is clear that there is a close relationship with social engineering, with some researchers describing phishing as a form of social engineering whereas others see phishing as employing social engineering. Regardless of the details, a unifying thread is the exploitation of human vulnerability for the creation of some form of deception for nefarious purposes.

Methods of defending against social engineering attacks are discussed in \cite{Ghafir}, \cite{Saleem} and \cite{Fan}. A clearly defined and well documented security policy is regarded as essential in protecting information infrastructure including robust user authentication \cite{Ghafir}. Two factor authentication is widely deployed and believed to be much more impervious to attack than simple passwords and PINS. However two factor authentication is itself vulnerable to social engineering \cite{Grimes} \cite{Siadati} by essentially tricking the victim into providing the verification code to the social engineer using the type of technique described in \cite{Grimes} and \cite{Siadati}.

This illustrates the inescapable fact that the central goal of social engineering is to evade whatever security mechanisms are deployed to  protect the target system by exploiting human vulnerabilities \cite{Saleem} \cite{Fan} \cite{Indrajit}.  Grimes \cite{Grimes} observes, from an industry perspective, that ``$\cdots$ even the very best technical controls will allow some amount of social engineering to get by your defenses. What will save you then is good security awareness training.'' Indeed \cite{Saleem} cites case studies which ``$\cdots$ demonstrate that security awareness is the crucial and most effective tool in the fight against social engineering attacks $\cdots$''.

Security awareness training is discussed from various points of view in \cite{Saleem}, \cite{Fan}, \cite{B&H} and \cite{Grimes} with the caution that this needs to be ongoing to overcome the natural tendency to forget or habituate \cite{ Saleem}. Resistance training, ``$\cdots$ aimed at making employees resilient against persuasion techniques that a social engineer may employ''  \cite{Ghafir} are also advocated by some, particularly where the interaction is in the form of a dialog. 

Extreme types of security awareness are proposals, discussed in \cite{Fan} and \cite{B&H}, for a system to monitor the emotional state of, for example, call centre operators, to ensure they are in the correct frame of mind to resist social engineering. However we regard this type of intrusive approach as draconian and likely to be ineffective because it would generate resentment and encourage evasion by the operator as well as being disruptive when the monitor decreed the operator to be in the wrong frame of mind. 

A more benign and, we suggest, more effective approach is to provide potential victims with an automated assistant which monitors the interaction with an attacker to detect indications that a social engineering attempt may be underway.  For example, Bhakta  and Harris \cite{B&H} describe a system for analysing each line in a text based interaction to detect instances of a predefined `topic black list' warning the potential victim and perhaps taking preemptive action to prevent a security violation. Similarly, in this paper we accept that the vulnerability of human operators to social engineering is inevitable and describe a system to assist them is assessing the overall reliability of the information they are being presented with.

\subsection{The Nature of Social Engineering}
\label{nature}

The centrality of human vulnerability in social engineering is emphasised by the philosophically well grounded definition proposed in  \cite{Wang} Section IV:
\begin{definition}
\label{def1}
\begin{quotation}
``In the context of cybersecurity, social engineering is a type of attack wherein the attacker(s) exploit human vulnerabilities by means of social interaction to breach cyber security, with or without the use of technical means and technical vulnerabilities.''
\end{quotation}
\end{definition} 

Continuing research in social engineering requires a sound intellectual foundation, beginning with a consistent ontology which then provides the basis for the underlying conceptual scheme. This, in turn, enables the empirical knowledge base to be ordered by constructing more specialised taxonomies. 

Previous work on developing conceptual models and taxonomies is reviewed in \cite{Fan}. Recently, Wang et. al. \cite{Wang2} have proposed a comprehensive foundation for social engineering research comprising an ontology and conceptual model with consequent taxonomies. Their conceptual model has eleven components, \cite{Wang2} Fig. 5, with each of these components having their own taxonomies populated by the ontology. 

Defining a set of formal relations between the components of the conceptual model leads to a knowledge graph in which the vertices represent the elements of the ontology and the edges represent the relations between the conceptual model components whose taxonomies include the elements.  All of this is machine readable so that the knowledge graph can be constructed computationally from empirical data on actual social engineering attack scenarios and then analysed to reveal general features of attacks.

For example the vertex degree, i.e. the number of edges associated with a vertex, is an indicator of how frequently an ontological element features in attacks. A relatively high degree can reveal a particular human or system vulnerability so that an alarm can be generated when it occurs during an attack or special attention can be paid to it in security awareness training \cite{Wang2}.

Human vulnerability is examined in detail in \cite{Fan} which proposes an ``I-E Model of Human Weakness"  as a basis for investigating social engineering attacks and potential defences. Here, `I' denotes internal which refers to the psychological characteristics of potential victims whereas `E' denotes external referring to the social environment in which manipulation of the victim occurs. Whereas \cite{Fan} uses the term `socio - psychology' as a general description of this class of model, we prefer the term `psycho - social' because it is consistent with the description of similar issues in the health sciences \cite{psycho} where the internal-external dichotomy is also employed, although the social environment there is somewhat different.

Because ``$\cdots$ the success of social engineering relies heavily on the information gathered $\cdots$'' \cite{Wang2} one of the conceptual model components  is  `social engineering information'. The importance of information is also stressed in \cite{Ghafir} and \cite{phishing_attacks}, where `information gathering' is listed as the first step in the attack cycle, and in \cite{Indrajit} where `information requirements analysis' is described as one of the prerequisites of an attack.

\section{Method}
\subsection{The Dialog Based Social Engineering Context}
\label{context}

While we agree with the generality of Definition \ref{def1}, in this paper we focus on that class of cyber security breach where the social engineer manipulates the target of an attack into revealing confidential information to be used in perpetrating some type of fraud. We distinguish between the target of the social engineering attack, on the one hand, and the victim of the subsequent fraud on the other. Whereas the target is a single human entity, the victim may be a corporate entity or the many customers of that corporate entity. 

In a typical phishing attack the target is also the victim \cite{phishing_attacks}. However in another class of attack, to which call centres in particular are susceptible, the social engineer attempts to persuade the target that the social engineer is, or represents, the intended victim by entering into some form of dialog. The immediate goal is to convince the target that the attacker is entitled to privileged information, information which can range from a simple account password to commercial or national secrets. If convinced, the target reveals the information to the attacker who then uses it to perpetrate some form of fraud against the victim. Identity theft, a growing problem, is a particular case of this form of impersonation.

Clearly the social engineer has to create a false identity based on some form of profile which will establish the credibility of the social engineer as the source of the information on which the attack is based. The wider question of assessing credibility in social networks has been surveyed in \cite{credibility} which asserts that credibility is a synonym for believability and trustworthiness although the authors concede the multidimensional nature of credibility in the social network context. Multidimensionality is also stressed in \cite{Wathen&Burkell} where trustworthiness and expertise are described as primary markers of source credibility. The basis of judgements of source credibility made by Facebook users is examined in \cite{Algarni} in four dimensions, one of which is sincerity. Judgements in this dimension were made in part on the basis of the source's profile indicating the importance of that profile being sufficiently detailed to be convincing.  

Source credibility has been examined in a variety of contexts, e.g. investigations of the evaluation of real and fake news articles \cite{real_fake} and the adoption of information in on-line communities \cite{Zang_Watts}. In both cases, the theoretical context is provided by Dual-Process models of human information processing. Whereas the Heuristic-Systematic variant of the Dual-Process model is employed in \cite{Zang_Watts}, the work in \cite{real_fake} is based on another variant, the Elaboration Likelihood Model, originally developed to investigate persuasion. 

The investigation of victimisation in phishing attacks in \cite{Luo} is also based on the Heuristic-Systematic Model (HSM) where information processing occurs in two modes. In the heuristic mode a message is evaluated quickly on the basis of a set of available cues, one of which is source credibility, whereas in the systematic mode the evaluation involves cognitive processes and is consequently much slower and more deliberative. 

These two modes are not mutually exclusive and have been demonstrated to operate concurrently \cite{Zang_Watts} \cite{Luo} such that one can moderate the other i.e. one can reinforce the other (referred to as additivity) or one can attenuate the effects of the other. In particular, source credibility has been demonstrated to reinforce the results of systematic processing \cite{Zang_Watts} and, conversely, systematic processing can attenuate the effect of heuristic cues such as source credibility.

In the type of phishing attack considered here, the social engineer presents an argument aimed at persuading the target to perform some action - supplying confidential information or logging into a fake website for example. The target invokes some form of heuristic-systematic process to evaluate the argument. However if systematic processing dominates, the flaws in the argument are likely to be exposed \cite{Luo}  so it is in the social engineer's interest to suppress systematic processing in favour of heuristic processing by manipulation of the target.  One method of achieving this is to place the target under time pressure which suppresses systematic processing \cite{Luo} \cite{Zang_Watts}.

An important aspect of the HSM is that it recognises that people will terminate the process when they feel comfortable with the judgement that has been evolving. In the model this is referred to as reaching the individual's {\it sufficiency threshold} which depends on the prevailing circumstances. Lowering the sufficiency threshold will result in the individual favouring the heuristic  over the systematic mode. 

In \cite{Luo}, it is proposed that the success of an attack depends, at least in part, on using an attack pretext to suppress the target's sufficiency threshold. Pretexting \cite{Luo} \cite{Wang} is where the social engineer, in preparing the attack, creates a scenario which provides the informational context within which the attack takes place, a context which can be also designed to reinforce the perception of source credibility. Indeed we regard the pretext as a core component of the type of dialog based phishing attack considered in this paper.

\subsection{Information Insufficiency and the Socially Engineered Phishing Attack Model}
\label{SEP}

However we believe that a more detailed account of the relationship between the pretext and the sufficiency threshold can be developed from the Risk Information and Seeking Processing (RISP) model of information science which incorporates the Heuristic-Systematic Model. The RISP Model has been widely researched in the information sciences as an explanation for how and why people seek additional information in the context of risk (see \cite{Yang} for a review and discussion). While developing the full relationship between social engineering and RISP is beyond the scope of this paper, the essential idea for our purposes is as follows.

Firstly we observe that risk is inherent in a dialog based phishing attack and that risk is perceived at some level by the target. One of the main components of RISP is  {\it information insufficiency} which is defined as the difference between an individual's current knowledge level and their sufficiency threshold. A higher level of information insufficiency, all else being equal, will drive the individual towards higher levels of information seeking and processing \cite{GDN} and, consequently, to bias their processing mode towards systematic processing.

In light of this we propose that pretexting, rather than lowering the sufficiency threshold as hypothesised in \cite{Luo}, instead raises the perceived level of current knowledge relevant to the attack. The pretext thus decreases the level of information insufficiency thereby reducing the perceived need to seek additional information by interaction with the attacker. Furthermore a carefully crafted pretext can lower the level to the point where the heuristics will override any doubts that may arise from systematic processing of the pretext. Consequently pretexting complements source credibility by ensuring the target is more strongly influenced by source credibility and other heuristics than by the merits of the arguments presented by the social engineer.

All of these factors lead us to define a model of a socially engineered phishing (SEP) attack, which we refer to as a SEP attack model, as: 

\begin{definition}
\label{SEPdef}
The SEP attack model has the following three components:
\begin{enumerate}
\item the source i.e. the attacker, with an associated source credibility,
\item the pretext providing the informational context in which the attack occurs and
\item the set of demands and the argument for why the target should conform to them.

\end{enumerate}
\end{definition}
Note that all three components are fake but presumably sufficiently plausible that the target is willing to engage with the attacker. 

Credibility generally increases with familiarity \cite{Wathen&Burkell} which implies that a completely unfamiliar pretext will be regarded with considerable skepticism by the target suggesting the engagement of some level of systematic processing. Provisionally accepting the pretext implies that at least some of the components of the pretext  are held by the target to be true. Because a high level of systematic processing would require that the large proportion of the components be held to be true, suppressing the sufficiency threshold reduces the proportion of components that the target needs to recognise as true in order to accept the pretext. 

The social engineer is then free to invent the remaining components but with the constraint that consistency is important for credibility \cite{Wathen&Burkell} so that  inconsistency in the information presented to the target results in a greater reliance on systematic processing \cite{Zang_Watts}. This clearly conflicts with the original aim of ensuring the dominance of heuristic processing. Consequently, the fake components of the pretext or message must be chosen to appear consistent with the components that the social engineer can assume the target will regard as true on the basis of the social engineer's research into the target's background.

All of this emphasises the critical role played by information as was stressed in Section \ref{nature}. It becomes clear that a social engineer needs to mount a significant effort to painstakingly acquire the information required for a successful attack. The nature of the information is such that the acquisition process can be indirect and complex so that social engineers ``$\cdots$ will tie little pieces of information they have acquired over time, decipher cues and signals given to them by multiple employees, and then connect the pieces of the jigsaw puzzle to unearth the information they have been after'' \cite{Saleem}.

Previous work has focussed on the exploitation of whatever factual information the social engineer can extract directly from available sources. However we argue that the `pieces of the jigsaw puzzle' acquired by a social engineer will in general be incomplete and perhaps inconsistent so that `unearth(ing) the information' will require significantly more that merely assembling a collection of facts. This is particularly the case with the type of personal detail that will lend a strong sense of authenticity to the pretext.

Specifically, we base this paper on our assertion that information acquisition can be enhanced by inferring additional information from, not only that already acquired but, importantly, contextual  information as well. Rarely, however, will this be a straightforward matter of deductive inference leading to certainty. Instead the inferences will be probabilistic in nature. 

\subsection{Probabilistic Inference, Analogical Inference, and Social Engineering}
\label{inference}

Probabilistic inference techniques are finding increasing application in database security research particularly in characterising and detecting network based attacks. In fact in a database inference attack, inference procedures are an intrinsic component of the attack methodology itself and developing effective countermeasures against inference attacks is an established aspect of database management research \cite{Farkas}. The underlying concept is that of an `inference channel' which is generated by the patterns of association in the data which support inferences. 

Of greater relevance to this paper are a number of techniques which have been developed to infer unknown (latent) attribute or demographic values from a known set derived from social network data. For example a machine learning method based on feature vectors is described in \cite{Zamal} whereas probabilistic graph theoretic techniques based on Markov Random Field concepts are proposed in \cite{MIT} and \cite{Iowa}. 

Underlying the probabilistic approaches is the characterisation of the set of attributes as random variables described by a multivariate probability distribution which is estimated, at least implicitly, from the known values. However in both \cite{MIT} and \cite{Iowa} the computational problem is made tractable by using binary variables and by adopting a pairwise Markov Random Field Model. The latter places quite strong constraints on the form of the multivariate distribution restricting it to a product of univariate and bivariate `potential' functions \cite{Koller}. In addition, the assortativity \cite{MIT} of the associated social network defines a neighbourhood structure which further restricts the form of distribution with consequent restrictions on the conditional independencies between the attributes \cite{Koller}. 

Even with their simplifying assumptions these are not simple, straightforward procedures either conceptually or computationally. To be clear, we are not suggesting that social engineers in general have the capacity to utilise these techniques to infer missing information even where the data was available. Instead we develop the argument below that the social engineer will make these inferences analogically based on a broad understanding of the socioeconomic context and human nature.

To formalise these ideas we use the terminology of Clarke \cite{Clarke}, which distinguishes between an entity and an identity. An entity is something in the world, not necessarily human, which can present many different identities depending on the particular role that entity is playing in the relevant context. Whereas the social engineer's intended victim is an entity, it is that entity's role in transacting with the organisation represented by the human target that is the social engineer's focus. 

Both the social engineer and the target are dealing with the particular identity presented by the entity  comprising that subset of the entities attributes required for the role.  Identities are distinguished by the values that those attributes have in each particular case, those values providing the basis for identification.

For identification purposes, the target will have access to a record associated with the victim in the form of a digital persona which Clarke defines as ``(t)he collection of data stored in a record (which) is designed to be rich enough to provide the record-holder with an adequate image of the represented entity or identity" \cite{Clarke}. Bearing in mind that the digital persona is, by design, a limited view of the victim, the key aspect of this is the social engineer having a sufficiently plausible  profile of the victim to convince the target that the social engineer is indeed the victim. 

Unfortunately this is made increasingly easy by the eagerness of people in general to place personal details in specific profiles on the various forms of social media. However, even hiding identities behind handles is an increasingly flimsy defence \cite{Wang} so that a social engineer is increasingly able to acquire key elements of the identity relevant to the attack. 

An important component of the social engineer's skill set is collecting these identity elements from a variety of public sources. Nevertheless, in general, publicly available information will be insufficient for plausibility, requiring supplementary information to flesh out the persona giving the target the sense that they are indeed communicating with the victim. For this to succeed the supplementary information need not be accurate but merely consistent with the digital persona. 

One of the key proposals of this paper is that this supplementary information can be derived by the social engineer, in part by inference from already obtained information perhaps by formal statistical techniques or, more than likely, by using general background knowledge of existing socio-economic patterns. Indeed we assume that astute social engineers are adept at observing and exploiting patterns in the socio-economic events surrounding them. Experimental psychology has demonstrated a clear connection between explanation, on the one hand, and learning and inference on the other \cite{Lombrozo2011}\cite{Lombrozo2016} so we assume, in addition, that those patterns which have an explanatory context, will predominate. 

The implication is that the social engineer makes inferences in an explanatory framework.  Based on the close psychological connection between explanation and analogy \cite{HLB}, we believe that this is best categorised as a form of analogical inference. For example a young person living in a country town is in an analogous situation to many young people living in many country towns. It can then be inferred that a particular young person is very likely to be unemployed {\it because} the economies of many country towns are contracting offering limited employment opportunities to young people as is well known. 

The essential point is that it is this analogous explanatory framework which lends plausibility, in the mind of the target, to the inferred information used in perpetrating the fraud rather than the formal statistical properties of the publicly available information. Consequently we do not require social engineers to have access to a formal database, nor to have the skills to analyse the data even if they did. Nevertheless the social engineer will be well aware informally of those statistical properties which will form part of the motivation for drawing  the inferences depending on the perceived level of statistical relevance.

The notion of statistical relevance in its various forms has played a role in the philosophical analysis of explanation in a scientific context since \cite{SalmonSR}. In experimental psychology the relationship between explanatory power and statistical relevance has been demonstrated using simple measures,\cite{colombo1}. At a more theoretical level, Schupbach and Sprenger \cite{S&S} characterise statistical relevance in terms of the degree to which one proposition H (e.g. a hypothesis) makes a second proposition E (e.g. empirical evidence) less surprising. They then use this in a Bayesian context as part of their definition of a measure of explanatory power which Schupbach \cite{Schup} has demonstrated successfully accounts for the judgements of explanatory power made by subjects in a series of psychological experiments. 

All of this suggests that astute social engineers will generally  base their analogical inferences on informal assessments of statistical relevance derived from observation of socio-economic patterns. However, instead of attempting to describe and analyse analogical inferences, we focus on the more easily defined task of analysing the underlying probabilistic structure that the psychological evidence suggests reflects and supports analogical inference. While the explanatory aspects of the type of analogical inference referred to here have a strong causal flavour, we will regard probabilistic information as purely associative having no sense of causality \cite{Pearl}.  However we recognise that this creates the potential for false positives.

By analysing the probabilistic structure of the publicly available information we can anticipate the drawing of analogical inferences and a key outcome of our research is an indicator of the degree to which a dataset is susceptible to inference. This indicator can be exploited to design countermeasures that reveal, and warn of, the patterns of association that underly the analogical inference, alerting potential targets to the possibility that information that they might be presented with is unreliable. 
\subsection{Contributions of the Paper}
\label{contributions}

We have discussed the critical role played by information in Sections \ref{nature} and \ref{SEP} with an essential part of the attack preparation being the careful assembly of as much detail about the attack victim and context as is feasible \cite{Wang2} \cite{Ghafir} \cite{phishing_attacks} \cite{Indrajit}. This information is used the construct the pretext which is an essential component of a social engineering phishing attack. The Heuristic-Systematic (HS) model of information processing developed in the Information Sciences has previously been used in \cite{Luo} as a basis for analysing the interaction between attacker and target. 

However, because the attack occurs in the context of risk, in this paper we extend that analysis by adopting the Risk Information Seeking and Processing (RISP) model \cite{GDN} which incorporates the HS model. An important component of the RISP model is information insufficiency which acts as the primary motivator for seeking and processing further information. In Section \ref{SEP} we invoke the RISP model to argue the importance of a sufficiently convincing pretext to reduce the target's information insufficiency level to the point where the target avoids the systematic mode of processing. The target then does not subject the available information to the more detailed analysis which might raise doubts in the target's mind about the veracity of the attacker's claims.  

Our underlying thesis is that the details revealed by the attacker's careful assembly of existing information are generally insufficient for building a convincing pretext. However they can be sufficient for the social engineer to infer probabilistically enough additional details to persuade the target that the claims being made in the pretext are valid if the target processes the claims heuristically. Consequently a seemingly convincing pretext may be less substantial than it appears because it is based partly on inferred attribute values. This raises the question of whether this situation can be detected and, if so, can the degree to which a pretext is dependent on inference be estimated. We seek to answer this question in this paper.

Our primary goal, then, is to devise a means of advising potential targets, particularly in informal database environments, that certain combinations of attributes might be unreliable for identification because the values of some are inferable from others. In other words the social engineer is likely to have inferred some attribute values thus appearing to be the victim who of course would know those values. We believe that our approach of using indications of inference as the basis for advising targets of the need for caution has not previously been reported in the literature.

Note that this is quite distinct from describing methods for actually drawing inferences, using, for example, the conceptually and computationally complex processes described in \cite{Zamal}, \cite{MIT}, and \cite{Iowa}. Instead, in Section \ref{contingency} below we develop a method of estimating the potential of a data set for inference to be performed without actually computing the inferences thus avoiding much of the complexity.

In the following we proceed from the observation that inference relies on relationships between the attributes where, in this case, the relationships are not only explanatory, thus supporting analogical inference, but are also manifested statistically enabling complementary probabilistic inference. The principal theoretical focus of the paper, then, is on the identification of attributes where unknown attribute values can be estimated by probabilistic inference from some database using the known values of other attributes. 

Like \cite{Guarnieri}, and \cite{epistemic}  we assume the existence of a database derived from some defined population from which can be extracted a multidimensional contingency table ${\mathcal T}$ in which the cells of the table are indexed by an ordered index vector $\mathbf{I}$. Each cell contains the number of members of the population with the combination of attribute values referenced by the index so the table can be conceptualised as an unnormalised discrete multivariate probability distribution.  However, instead of using the posterior joint distribution \cite{epistemic} directly or the posterior marginal distribution \cite{Guarnieri}, we employ a log linear transformation of the prior distribution to analyse the probabilistic relations between subsets of attributes which give rise to the structure in the table.  

Because inferences are drawn from prior information we focus on the contingency table equivalent of a conditional probability distribution which we refer to as a conditional subtable. We will assume inference is immediate in the sense that the conditional distribution is dominated by, i.e. concentrated on, a particular combination of attribute values. These can then be taken as the unknown values with high likelihood. 

We develop a novel computational procedure for detecting particular subsets of attributes which have conditional subtables dominated by a small set of values without having to analyse each individual subtable. The basis of this procedure is the observation that a uniform distribution is completely uninformative and that, in order for probabilistic inference to be viable, a distribution must exhibit significant deviation from uniformity. This is because of  some property of whatever mechanism is generating the cell contents - a mechanism that causes some cells to be significantly more probable than others. 

In other words these cells are salient and we measure salience as a deviation from uniformity. Roughly speaking we argue that a high level of salience in a subset of attributes results in the multivariate conditional subtable of those attributes being concentrated on just some small number of the attribute levels in the subtable. This does not require that the dominant cells have any topological relationship. In particular it does not require that they be neighbours so that standard measures based, for example, on second moments, are not suitable for our purpose. 

In Section \ref{loglinear} we develop the vector space theory underlying our analysis utilising the orthogonal log linear transformation introduced by Dahinden et. al. \cite{Dar2}. Importantly this is based on projecting the vector representing the logarithm of the table onto a subspace orthogonal to the uniform vector. After analysing the orthogonalised log linear design matrix we state and prove Theorem \ref{Theorem1} in Section \ref{Conditional}. This novel result enables us to reduce a large contingency table to one involving a subset of attributes using geometrical means as a form of marginalisation while preserving salience as much as possible. 

Theorem \ref{Theorem1} also provides the basis for our definition of a novel indicator, Probabilistic Salience, of the degree to which statistical data can support an inference that particular attributes have specific values (Definition \ref{salience_factor}). We describe a means of evaluating that indicator from the data and demonstrate how it can identify those subsets of attributes in a contingency table which are most vulnerable to inferencing . Then, in Section \ref{salience}, we show why Probabilistic Salience does indeed measure the degree to which a subset of attributes is dominated by a relatively small number of members.

Most of Section \ref{xi} as well as the rest of Section \ref{loglinear} report the results of research which, to the best of our knowledge, is original.  We believe that this work also makes a significant contribution to the contingency table literature because we were unable to find anything there that investigated this type of problem. 

We envisage a novel system which employs Probabilistic Salience to generate a warning when some of the attribute values being claimed by a social engineer might have been inferred from others without having to classify them. The warning would take account of the level of any formal authentication provided and could be as simple as a red, amber, green traffic light indication. While this system would be most easily implemented in a text based dialog we see nothing in principle which would prevent speech recognition being used in a verbal dialog. Such a system could complement the topic blacklist system described in \cite{B&H} as well as that referred to in \cite{Wang2} and might form a component of a more comprehensive intelligent assistant for combatting social engineering attacks in real time.

The central objective of our proposed probabilistic salience based warning system is to increase the target's information insufficiency level by inducing the target to derate the inferred attributes thereby reducing the target's current knowledge level. If the inferred content is large enough, the information insufficiency level will increase to the point where the target transitions to an information processing mode dominated by systematic processing. 

In this mode the target could be expected to take a more deliberative and analytic approach to evaluating the attacker's claims and expose the inevitable flaws and inconsistencies leading to a rejection of the attack. It is not intended that the system directly intervene in the attack by, for example, preemptively terminating the contact.  

Investigating the conditions under which probabilistic inferencing  can be performed, and assessing the potential for probabilistic inferences to be drawn is, to the best of our knowledge, a novel approach which this paper brings to the understanding of social engineering. Furthermore this novel approach is able to draw upon the understanding of human information processing being developed in the information sciences to begin constructing techniques for actively counteracting social engineering. We suggest that this is an example of a potentially fertile area of research applying the characteristics of human information processing in the context of risk to the development of effective countermeasures to social engineering.

Finally in Section \ref{discussion} we propose a potential application of Probabilistic Salience and Theorem \ref{Theorem1} in what we refer to as `de-personalisation' of a data set by analogy with the de-identification techniques used in data privacy.  The objective here is to inhibit probabilistic inference by selectively reducing the data's Probabilistic Salience without excessively compromising its utility. 

\section{Theory}
\subsection{The Multidimensional Contingency Table}
\label{contingency}
The general problem area is that of $N$ dimensional contingency tables \cite{algebraic} formed from an ordered set of categorical variables, i.e. attributes, $C = \{C_{N-1}, \cdots C_{0}\}$. Each variable consists of a set of values or levels  such that, for example, the $k$th variable, $C_k$, comprises the $M_k$ levels $\{a_{k1},a_{k2} \cdots a_{kM_k}\}$.  There are $M_{T}=\prod_{k=1}^{N}M_{k}$ cells in the table, each being identified by an ordered index set 
\[\boldsymbol{\mathfrak{i}} = \langle \mathfrak{i}_{N-1},\mathfrak{i}_{N-2} \cdots \mathfrak{i}_0 \rangle\] 
where $\mathfrak{i}_k \in \{0\cdots M_k-1\}$ so that each attribute $C_k$ is directly associated with the index $\mathfrak{i}_k$. Each cell contains the number of members of the population which exhibit the joint occurrence of the attribute levels. However, for notational convenience in what follows, there will be the same number of levels in each attribute, $M$, so that $M_{T}=M^{N}$. 

We regard the table as a geometric entity, in this case a hypercube of cells. To provide flexibility in associating an $N$ element vector of attribute values with the indices of the corresponding cell, we  adopt the convention of representing the $k$th `axis' of the hypercube by the $N$ element standard basis vector $\bf{e}_k$ in which the $k$th element is `$1$' with the remainder zero. This enables the set of attributes to be referenced by the index matrix 
\begin{equation}
  \mathbf{I} = \left | \mathbf{e}_{N\!-\!1} \cdots \mathbf{e}_0 \right |.   \label{indexM}
\end{equation}
If the levels of the attributes designating the $\ell$th cell are collected in the vector $\mathbf{a}_\ell = \langle a_{\ell_{N\!-\!1}} \cdots \mathbf{a}_{\ell_0}\rangle$, $\mathbf{a}_{\ell_j} \in \{0\cdots M_k-1\}$, the indices of the cell are given by 
\begin{equation}
 \boldsymbol{\mathfrak{i}}_\ell = \mathbf{I}\:\mathbf{a}_{\ell}. \label{indexv}
 \end{equation}
 
 The index matrix construct (\ref{indexM}) also provides a mechanism for dealing with subtables. If $C_{\ell_k} \cdots C_{\ell_1}$ is some subset $\mathbf{\mathsf{c}}^\ell$ of $k$ attributes, the attribute index vector $\mathbf{i}^s_\mathbf{\mathsf{c}}$ is
 \[\mathbf{i}^s_\mathbf{\mathsf{c}} = \langle i_{\ell_k} \cdots i_{\ell_1}\rangle\]
  where there is no implication that the integers $\{\ell_k \cdots \ell_1\}$ are contiguous.  Using $\mathbf{i}_\mathbf{\mathsf{c}}$ as a key provides the subtable index matrix
  \begin{equation}
  \mathbf{I^s} = \left | \mathbf{e}_{\ell_k} \cdots\mathbf{e}_{\ell_1} \right |. \label{indexMs}
  \end{equation}

The contingency table is assumed to be a member of an ensemble of similar tables with the common characteristic that they represent the same population so that their respective entries sum to the population size say $N_T$. More formally this is equivalent to assuming that each table is drawn from a multinomial distribution although we will not make use of this fact. These tables can be represented as vectors by ordering the table indices according to some rule. Here we choose a lexicographic ordering in which the index, $\mathfrak{i}_{0}$ varies the most rapidly followed by $\mathfrak{i}_{1}$ and so on resulting in a table vector, $\mathcal{T}\in \mathfrak{R}^{M_T}$. 

Much of the literature on contingency tables is concerned with the estimation of the cell contents from incomplete data. However, to avoid being distracted by estimation procedures, important though they are, we will assume that the data in the cells is a reliable representation of the population.  The only complication is presented by zero entries in one or more cells but we will take a pragmatic approach here and replace zeros by ones by applying an affine transformation to the table with minimal impact on the conclusions for reasons that will become apparent below. 

\subsection{The Log Linear Transformation}
\label{loglinear}

The log linear transformation is based on the cell by cell logarithmic transformation of a contingency table vector  ${\mathcal T}=[\rho_0, \cdots \rho_{M_T}]^\mathsf{T}$ to generate a new table  $\mathbf{T}=[\gamma_0 \cdots \gamma_{M_T}]^\mathsf{T}$ where $\gamma = \log \rho$. In its general form the log linear transformation is (\cite{DLS})
\begin{equation}
\mathbf{T}(\boldsymbol{\mathfrak{i}})  = \sum_{\mathbf{\mathsf{c}}{\subseteq C}} \xi_{\mathbf{\mathsf{c}}}(\mathbf{i}_{\mathbf{\mathsf{c}}}) = \xi_{0}+\sum_{\ell=1}^{N}\xi_{\ell}+\!\!\sum_{\stackrel{\scriptstyle{\mathsf{c}}\subseteq C}{{\scriptstyle |\mathsf{c}|=2}}} \!\!\xi_{{\mathbf{\mathsf{c}}}}(\mathbf{i}_\mathsf{c})+\!\! \sum_{\stackrel{\scriptstyle{\mathsf{c}}\subseteq C}{{\scriptstyle |\mathsf{c}|=3}}} \!\!\xi_{{\mathbf{\mathsf{c}}}}(\mathbf{i}_\mathsf{c})\cdots   \label{xifunct}
\end{equation}
where the $\xi$ functions are parameterised by a subset of the attributes.  

In what follows it is essential to impose some form of ordering on the set of subsets $\mathbf{\mathsf{c}}$ i.e. on the power set of $C$, $\mathbb{P}(C)$. In the ordering adopted here, the first term has $|\mathbf{\mathsf{c}}| = 0$, i.e. the null set representing the constant term. Next comes the group of subsets with $|\mathbf{\mathsf{c}}|=1$ representing the main effects , then those with $|\mathbf{\mathsf{c}}|=2$  and so on. The terms with $|\mathbf{\mathsf{c}}|\ge 2$ are referred to as `interaction terms' because they represent probabilistic interactions between the members of the subset \cite{DLS}\cite{Koller}, beginning with first order interactions between pairs of attributes. Within each group the subsets are ordered following the lexicographic principle applied to their member's indices. For example where $|{\mathbf{\mathsf{c}}}|=2$, i.e. first order interactions, the ordering is 
\begin{equation}
C_{N\!-\!1}C_{N\!-\!2},\ C_{N\!-\!1}C_{N\!-\!3} \cdots C_{N\!-\!1}C_{0},\ C_{N\!-\!2}C_{N\!-\!3} \cdots C_1C_0.     \label{pairwise}
\end{equation}

Clearly, defining the $\xi$ functions, which are in the form of an $M_{T}$ dimensional vector, is a critical  task in constructing the expansion. The key development on which the process described here is constructed has been introduced by Dahinden et. al. \cite{Dar2} where they express each $\xi$ function as a linear combination of a set of basis vectors which span the subspace containing it. There is one basis vector for each particular combination of the levels associated with a subset of attributes.  The number of such basis vectors, i.e. the dimension of the subspace, is determined by the number of possible configurations of the levels which is $M^{|\mathbf{\mathsf{c}}|}$. Each $\xi$ vector can then be obtained from the data itself by determining the coefficients of the linear combination.

In matrix form, the log-linear model is 
\begin{equation}
 {\bf T} = {\bf X} {\boldsymbol \alpha} \label{loglin}
\end{equation}
where ${\bf X}$ is a design matrix and ${\boldsymbol\alpha}$ is a vector of coefficients of the form
\[ \langle \alpha_0\  \alpha_1\  \alpha_2 \cdots \alpha_N\  \alpha_{11}\  \alpha_{12} \cdots \alpha_{NN}\  \alpha_{111}\  \alpha_{112} \cdots \alpha_{NNN} \cdots\rangle ^\mathsf{T}  \]

 The basis vectors of a particular $\xi$ function then form a  subset of the column vectors of $X$ i.e. a submatrix of $X$ where each submatrix represents one of the subsets, $\mathsf{c}$, of attributes involved in a particular level of interaction. Dahinden et. al. label the submatrices so generated as $ X^{C_{k1},C_{k2}\cdots C_{k|{\boldsymbol{\mathsf{c}}}|}} = X^{\mathsf{c}_k}$ where the elements of each submatrix are either zero or one. The algorithm for generating the basis vectors is given in \cite{Dar2} but the essential idea is described below. However, because the detailed form of the design matrix is a critical part of our development, we examine its construction in more detail in Section \ref{form}.

To proceed, it is necessary to introduce some additional notation. Firstly let $k$ be the interaction index, $k=|\boldsymbol{\mathsf{c}}|$, so that $k=0$ indexes the constant term in (\ref{xifunct}), $k=1$ the main effects and so on. For each value of $k$, the number of submatrices of the form $X^{\boldsymbol{\mathsf{c}}_{i}}$ is $N_{k}=\binom{N}{k}$ and the number of columns in each submatrix is $M^{k}$, one column for each of the possible combinations of levels in the set of $k$ attributes. Keep in mind that each cell in the table ${\mathcal T}$ is denoted by a specific index vector $\boldsymbol{\mathfrak{i}} = \langle \mathfrak{i}_{N-1} \cdots \mathfrak{i}_{j} \cdots \mathfrak{i}_0\rangle, \;\mathfrak{i}_{j}\in \{0\cdots M-1\}$ so that the $r$th component of the table vector $\mathbf{T}$ is denoted by
\begin{equation}
r = \sum_{j=0}^{(N-1)}\mathfrak{i}_jM^{j} = (\mathbf{R})_M.         \label{lex}
\end{equation}
where $(\mathbf{R})_M$ is the $M$-ary number $\boldsymbol{\mathfrak{i}}$ expressed in radix $M$ form.  Note that this is simply the operation of counting through the elements of the table vector with radix $M$.

 Consider one specific value of $k$ and one specific submatrix, $X^{\boldsymbol{\mathsf{c}}^{\ell}_{k}}$, where $\boldsymbol{\mathsf{c}}^{\ell}_{k}\in \boldsymbol{\mathsf{c}}_{k} \subseteq \mathbb{P}(C)\  s.t. \ |\boldsymbol{\mathsf{c}}_{k}| = k$ and $\ell \in \boldsymbol{\ell}=\{1 \cdots\binom{N}{k}\}$ which specifies the subsets $\boldsymbol{\mathsf{c}}^\ell_k$ of $\boldsymbol{\mathsf{c}}_k$. For each value of $\ell$, the attributes are designated by the ordered vector of attribute numbers $\mathbf{i}^{s_\ell}_k=  \langle i_{\ell_k}\cdots i_{\ell_1}\rangle $ so that e.g. $C_{i_{\ell_k}}$ is the leftmost attribute in the subset. Then each column of the submatrix represents one particular combination of the levels of the subset of attributes designated by $\mathbf{i}^{s_\ell}_k$. Note that there is no implication that the $k$ indices are contiguous which, in general, they are not as indicated by the example (\ref{pairwise}). 
 
 If the levels of the attributes are $a \in \{0 \cdots M-1\}$, the vector $\mathbf{a}^{k,\ell} = \langle a_{\ell_k} \cdots a_{\ell_1} \rangle$, represents one combination of values of the attributes  $\boldsymbol{\mathsf{c}}^\ell_k$. The $\nu$th column of the submatrix is then associated with a particular $\mathbf{a}^{k,\ell}_j$ as well as $\mathbf{i}^{s_\ell}_k $. The remaining ${(N-k)}$ table axes form the vector $\mathbf{i}^{g_\ell}_k$ with corresponding attributes vector $\mathbf{g}^{k,\ell} = \langle g_{\ell_{N-k}} \cdots g_{\ell_1}\rangle$. Using (\ref{lex}),  the elements of the $\nu$th column vector having $\mathbf{i}^{s_\ell}_k $ and $\mathbf{a}^{k,\ell}_j$ in common, $\{x\}_g$, are designated by the table index vector, from (\ref{indexv}) and (\ref{indexMs})
 \begin{equation}
 \boldsymbol{\mathfrak{i}}_j = \mathbf{I}^{s_\ell}_ k\mathbf{a}^{k,\ell}_j + \mathbf{I}^{g_\ell}_{N-k} \mathbf{g}^{k,\ell}   \label{column_index}
 \end{equation} 
 where the elements of $\mathbf{g}^{k,\ell}$ range over all $M^{N-k}$ combinations of the attribute levels.
 
To define the basis vector recognise that the available information does not allow the elements of $\{x\}_g$ to be distinguished. Consequently the $\nu$th column vector becomes a basis vector by setting the elements of $\{x\}_g$ to unity with the remaining elements zero.

The lexicographic ordering is a bijective mapping, so that if there is a one at a particular position in the $\nu$th column of the submatrix $X^{\boldsymbol{\mathsf{c}}^{\ell}_{k}}$, there cannot be a one in the corresponding position in any of the remaining $M^{k}-1$ columns. Consequently, the columns of the submatrix in question, i.e. the basis vectors, are orthogonal. Furthermore, because the total number of ones in the submatrix is $M^{k}M^{N-k}=M^{N}$, the sum $\sum_{j=1}^{M^{k}} {X_j^{\boldsymbol{\mathsf{c}}^{\ell}_{k}}}$ is equal to the unit vector $\widehat{X}^{0}$.

The full matrix then has 
\[1+NM + \cdots\binom{N}{k}M^{k} \cdots +M^{N} = (1+M)^{N}>M^{N}\]
 columns so that (\ref{loglin}) is very underdetermined. Conventionally the approach to this problem is to impose `identifiability' by removing the null space of $X$ using, for example, the Moore-Penrose pseudo inverse. However this approach results in losing the explicit connection between the coefficient vector and the interaction terms. 
 
 Fortunately Darhinden et. al., \cite{Dar2}, have derived an orthogonal form of $X$ which leads to a coefficient vector that can be related to the interaction terms in (\ref{xifunct}). The essential idea is as follows. Using a modified version of the notation in \cite{Dar2}, the vector space, $\widehat{X}^{0}$ spanned by the constant vector $k=0$ has dimension 1. Then the vector space $\widehat{X}^{C_{\ell}}$, spanned by the columns of $X^{C_{\ell}}$, is composed of $\widehat{X}^{C_{\ell}}_{-} = \widehat{X}^{0}$ and its orthogonal complement, $\widehat{X}^{C_{\ell}}_{\perp}$, in $\widehat{X}^{C_{\ell}}$ the orthogonal complement having dimension $M-1$. The space, $\widehat{X}^{C_{\ell_{1}}C_{\ell_{2}}}$, spanned by $X^{C_{\ell_{1}}C_{\ell_{2}}}$ is composed of ${\widehat X}^{C_{\ell_{1}}C_{\ell_{2}}}_{-} = \{\widehat{X}^{C_{\ell_{1}}},\widehat{X}^{C_{\ell_{2}}}\}$ and its orthogonal complement, ${\widehat X}^{C_{\ell_{1}}C_{\ell_{2}}}_{\perp}$ in ${\widehat X}^{C_{\ell_{1}}C_{\ell_{2}}}$ which has dimension $M^{2} - 2(M-1)-1  = (M-1)^{2}$.

For the $k$th interaction index, there are $\binom{N}{k}$ spaces of the form $\widehat{X}^{C_{\ell_{1}}\cdots C_{\ell_{k}}}$ where each of the $\binom{N}{k}$ sets $\{\ell_1 \cdots \ell_k\}$ is a distinct subset of the index set $\{1\cdots N\}$. The space spanned by the columns of $X^{\boldsymbol{\mathsf{c}}^{\ell}_{k}}$, $\widehat{X}^{\boldsymbol{\mathsf{c}}^{\ell}_{k}}$, is composed of the subspace 
\[ \widehat{X}^{\boldsymbol{\mathsf{c}}^{\ell}_{k}}_{-} = \{\widehat{X}^{C_{\ell_{1}}\cdots C_{\ell_{k-1}}},\widehat{X}^{C_{\ell_{1}}\cdots C_{\ell_{k-2}}C_{\ell_{k}}},\cdots \widehat{X}^{C_{\ell_{2}}\cdots C_{\ell_{k}}}\}\]
 and its orthogonal complement in $\widehat{X}^{C_{\ell_{1}}\cdots C_{\ell_{k}}}$, $\widehat{X}^{\boldsymbol{\mathsf{c}}^{\ell}_{k}}_{\perp}$. This orthogonal complement has dimension
 \begin{equation} 
 M^{k} - \left(1+ \binom{k}{1}(M-1) +  \cdots + \binom{k}{k-1}(M-1)^{k-1}\right) = (M-1)^{k}.  \label{OCdim}
\end{equation}

Beginning with the constant vector, the orthogonalised form of $\bf{X}$, ${\overline{\bf X}}$, is derived by successively projecting the columns of the submatrix $X^{\boldsymbol{\mathsf{c}}^{\ell}_{k}}$ onto the orthogonal complement $\widehat{X}^{\boldsymbol{\mathsf{c}}^{\ell}_{k}}_{\perp}$. Then ${\overline{\bf X}}$ is constructed from the constant vector and the columns of the submatrices following the ordering described above. Consequently the number of columns is
\begin{equation}
1 + N(M-1)  + \cdots +\binom{N}{k}(M-1)^{k} + \cdots  + (M-1)^{N} = M^{N} \label{matrix_size}
\end{equation}
so that ${\overline{\bf X}}$ is a square matrix.

\subsubsection{Constructing the $\xi$ Functions}
\label{xi}
In general, if the orthogonalised matrix is ${\overline{\bf X}}$ then the log-linear model becomes, from (\ref{loglin}),
\begin{equation}
{\bf T}  = {\overline{\bf X}}{\boldsymbol{\beta}}   \label{mloglin}
\end{equation}
where ${\boldsymbol{\beta}} $ is the modified vector of structural parameters.  The vector ${\bf T}$ therefore exists in a space 
\begin{equation}
{\widehat{\overline X}} = \widehat{X}^{0} + \sum_{k=1}^{N} \sum_{\ell = 1}^{\binom{N}{k}} \widehat{X}^{\boldsymbol{\mathsf{c}}^{\ell}_{k}}_{\perp}
\end{equation}
where $\boldsymbol{\mathsf{c}}^{\ell}_{k}$ is the $\ell$th subset of attributes having cardinality $k$ in the lexicographic ordering specified above.

Unfortunately, whereas the elements of ${\boldsymbol{\alpha}}$ in (\ref{loglin}) are associated directly with the corresponding interaction terms in ${\bf X}$ those in ${\boldsymbol{\beta}}$ are not. For example, the elements of ${\boldsymbol{\alpha}}$ corresponding to the first order interactions describe the structure of the interaction graph, while those associated with the first order interactions in ${\boldsymbol{\beta}}$ do not, at least not directly. This is because the coefficients associated with the $\overline{X}^{\boldsymbol{\mathsf{c}}^{\ell}_{2}}$ sub matrix of ${\overline{\bf X}}$ in (\ref{mloglin}) are the coefficients of the basis functions spanning the subspace $\widehat{X}^{\boldsymbol{\mathsf{c}}^{\ell}_{2}}_\perp$. 

However, because the matrix $\overline{\bf X}$ is partitioned into the submatrices, $\overline{X}^{\boldsymbol{\mathsf{c}}^{\ell}_{k}}$, whose columns are the orthogonal basis vectors of the subspaces $\widehat{X}^{\boldsymbol{\mathsf{c}}^{\ell}_{k}}_{\perp}$, the vector $\boldsymbol{\beta}$ can be partitioned into subvectors as
\begin{equation} 
\boldsymbol{\beta} = (\beta_{0} \; \boldsymbol{\beta}_{11} \cdots \boldsymbol{\beta}_{1N} \cdots\ \cdots \boldsymbol{\beta}_{k1} \cdots \boldsymbol{\beta}_{k\binom{N}{k}} \cdots\  \cdots) \label{bvector}
\end{equation}
with $\boldsymbol{\beta}_{k\ell} = (\beta_{k\ell 1}\; \cdots\; \beta_{k\ell (M-1)^{k}}).$
Then (\ref{mloglin}) becomes
\begin{equation}
{\bf T} = {\overline X}^{0}\beta_{0 } + \sum_{k=1}^{N}\sum_{\ell=1}^{\binom{N}{k}} \overline{X}^{\boldsymbol{\mathsf{c}}^{\ell}_{k}}\boldsymbol{\beta}_{k\ell}.    \label{orthexp}
\end{equation}

Now each term of the outer summation in (\ref{orthexp}) does refer to a specific interaction, namely, the interaction between the members of $\boldsymbol{\mathsf{c}}^{\ell}_{k}$.

To connect with the $\xi$ functions of (\ref{xifunct}) recognise that a vector $\boldsymbol{\chi}^\perp_{k\ell}$ representing the function $\xi_{\boldsymbol{\mathsf{c_{k}^{\ell}}}}:|\boldsymbol{\mathsf{c_{k}^{\ell}}}|=k,\ell \in \{1 \cdots \binom{N}{k}\}$ exists in the subspace spanned by the basis vectors forming the submatrix ${\overline X}^{\boldsymbol{\mathsf{c}}_{k}^\ell}$ of ${\overline{\bf X}}$ in (\ref{mloglin}).  Furthermore, the subvector ${\boldsymbol \beta_{kl}}$ is associated with $\xi_{\boldsymbol{\mathsf{c_{k}^{\ell}}}}$. In other words the function $\xi_{\boldsymbol{\mathsf{c_{k}^{\ell}}}}$ is represented by the linear combination
\begin{equation}
\boldsymbol{\chi}^\perp_{k\ell} = {\overline X}^{\boldsymbol{\mathsf{c}}_{k}^\ell}\boldsymbol{\beta}_{k\ell}.  \label{chivector}
\end{equation}
Because the underlying subspaces are orthogonal, $\boldsymbol{\chi}^\perp_{k\ell}$ is orthogonal to the vectors representing the other $\xi$ functions.  

Exploiting the fact that $\overline{{\bf X}}$ is an orthogonal matrix, the vector $\boldsymbol{\beta}$ is obtained by premultiplying ({\ref{mloglin}}) by $\overline{\bf X}^{\mathsf{T}}$ to give
\[\overline{\bf X}^{\mathsf{T}}{\bf T}  = \overline{\bf X}^{\mathsf{T}}{\overline{\bf X}}{\boldsymbol{\beta}}  = D_{\overline{\bf X}^{\mathsf{T}}}{\boldsymbol{\beta}} \]
where $D_{\overline{\bf X}^{\mathsf{T}}}$ is a diagonal matrix containing the squared magnitudes of the basis vectors and is readily inverted. Then
\begin{equation}
\boldsymbol{\beta} = D^{-1}_{\overline{\bf X}^{\mathsf{T}}}\overline{\bf X}^{\mathsf{T}}{\bf T} \label{betaV}
\end{equation}
and
\begin{equation}
\boldsymbol{\beta}_{k\ell} = D^{-1}_{\boldsymbol{\mathsf{c}}_{k}^{\ell}}{\overline X}^{{\boldsymbol{\mathsf{c}}_{k}^{\ell}}^{\mathsf{T}}}\bf{T} \label{betasubV}
\end{equation}
where $D_{\boldsymbol{\mathsf{c}}_{k}^{\ell}}$ is a diagonal matrix containing the squared magnitudes of the basis vectors in the subspace $\widehat{X}^{\boldsymbol{\mathsf{c}}^{\ell}_{k}}_{\perp}$.
From (\ref{chivector})
\begin{equation}
\boldsymbol{\chi}^\perp_{k\ell}= {\overline X}^{\boldsymbol{\mathsf{c}}_{k}^{\ell}} D^{-1}_{\boldsymbol{\mathsf{c}}_{k}^{\ell}}{\overline X}^{{\boldsymbol{\mathsf{c}}_{k}^{\ell}}^{\mathsf{T}}} \bf{T}  \label{projector}
\end{equation}
so that
\begin{eqnarray}
|\boldsymbol{\chi}^\perp_{k\ell}|^{2} & = &{\bf T}^{\mathsf{T}}{\overline X}^{\boldsymbol{\mathsf{c}}_{k}^\ell} D^{-1}_{\boldsymbol{\mathsf{c}}_{k}^{\ell}} {\overline X}^{{\boldsymbol{\mathsf{c}}_{k}^{\ell}}^{\mathsf{T}}} {\overline X}^{\boldsymbol{\mathsf{c}}_{k}^\ell}  D^{-1}_{\boldsymbol{\mathsf{c}}_{k}^{\ell}}{\overline X}^{{\boldsymbol{\mathsf{c}}_{k}^{\ell}}^{\mathsf{T}}}\bf{T}   \nonumber \\
& = & {\bf T}^{^{\mathsf{T}}}{\overline X}^{\boldsymbol{\mathsf{c}}_{k}^\ell} D^{-1}_{\boldsymbol{\mathsf{c}}_{k}^{\ell}}{\overline X}^{{\boldsymbol{\mathsf{c}}_{k}^{\ell}}^{\mathsf{T}}}\bf{T}      \label{magsq}
\end{eqnarray}
because of orthogonality. This shows that the matrix in (\ref{projector} is self adjoint and idempotent so that it is in fact a projector onto the subspace $\widehat{X}^{\boldsymbol{\mathsf{c}}^{\ell}_{k}}_{\perp}$. Consequently the magnitude of the function $\xi_{\boldsymbol{\mathsf{c_{k}^{\ell}}}}$  is the magnitude of the projection of $\bf{T}$ onto the subspace $\widehat{X}^{\boldsymbol{\mathsf{c}}^{\ell}_{k}}_{\perp}$ i.e. the subspace in which the vector representation of the function $\xi_{\boldsymbol{\mathsf{c}_{k}^{\ell}}}$ exists. In a sense it describes the magnitude of the contribution the population data in the table  makes to the interaction represented by $\xi_{\boldsymbol{\mathsf{c}_{k}^{\ell}}}$.

\subsection{The Structure of the Design Matrix}
\label{form}

To explore the overall form of the design matrix $\overline{X}$ we return to the construction of the original form $X$ and the discussion following (\ref{lex}). The columns of the submatrix $X^{C_{i_k},\cdots C_{i_1}}$ are numbered left to right by the radix $M$ number $\nu = (a_k, \cdots a_1)_M$ i.e. in the first column of the submatrix, $a_k=0\cdots a_1=0$ whereas in the second column $a_k=0\cdots a_2=0$ and $a_1=1$.  There are therefore $M^k$ columns in the submatrix where each is generated from the preceding column by counting in radix $M$.
\begin{lemma} 
Column $\nu = a_kM^{k-1}\cdots +a_2M+a_1$ of the submatrix $X^{C_{i_k}\cdots C_{i_1}}$ has the form
 \begin{equation}
 \begin{array}{ll}
 \mathtt{X}^{i_k \cdots i_1}(a_k\cdots a_1) = \\
 \left[\underbrace{\overbrace{0\cdots0}^{a_kM^{i_k}}\underbrace{\overbrace{0\cdots0}^{a_{k\!-\!1}M^{i_{k\!-\!1}}}\cdots\underbrace{\overbrace{0\cdots 0}^{a_1M^{i_1}}\overbrace{1\cdots1}^{M^{i_1}}\overbrace{0\ \cdots\ 0}^{(M\!-\!1\!-\!a_1)M^{i_1}}}_{\times M^{i_2-i_{1}-1}}\cdots\overbrace{0\ \cdots\ 0}^{(M\!-\!1\!-\!a_{k\!-\!1})M^{i_{k\!-\!1}}}}_{\times M^{i_k-i_{k\!-\!1}-1}}\overbrace{0\ \cdots\ 0}^{(M\!-\!1\!-\!a_k)M^{i_k}}}_{\times M^{N-i_k-1}} \right]^\mathsf{T}. \label{original_column}
 \end{array}
 \end{equation}
 \label{lemma1}
 \end{lemma}
 
 \begin{proof}
The column is initially all zeros, the table index is $\boldsymbol{\mathfrak{i}}=\langle 0 \cdots 0\rangle$ and the index of the table vector $\mathcal{T}$ is $r\!=\!(0 \cdots 0)_M\!=\!0$. Now begin incrementing the table index in a counting operation. A count of $a_1M^{i_1}$ results in $\mathfrak{i}_{i_1} = a_1$, a  count of $a_2M^{i_2}$ sets $\mathfrak{i}_{i_2}=a_2$ and so on until a final count of $a_kM^{k}$ will set $\mathfrak{i}_{i_k}=a_k$.  At this point the table index $\boldsymbol{\mathfrak{i}}$ points to the first cell where the attributes in $\boldsymbol{\mathsf{c}}^\ell_k$ have the values given by $\langle a_k \cdots a_1 \rangle$ and  a one is  entered in row $r=a_kM^{i_k}+\cdots +a_1M^{i_1}$ (remembering that the row numbering starts at zero) of column $\nu$. There are then a total of $r$ zeros preceding this one in the column.
 
 A further count of $M^{i_1}$ is required before $\mathfrak{i}_{i_1}$ goes to $(a_1+1)_M$ so that there are $M^{i_1}$ consecutive ones in the column. This is followed by a sequence of $(M-1-a_1)M^{i_1}$ zeros before all of the $(\mathfrak{i}_{i_1} \cdots \mathfrak{i}_0)$ indices are reset to zero. This whole first level pattern is then repeated until $\mathfrak{i}_{i_2}$ goes to $(a_2+1)_M$ i.e. the pattern of $a_1M^{i_1}$ zeros, $M^{i_1}$ ones and $(M\!-\!1\!-\!a_1)M^{i_1}$ zeros occurs $M^{(i_2-i_1-1)}$ times. 
 
 This set of repeated first level patterns extends the $a_2M^{i_2}$ zeros from the initial count and this second level pattern is completed by the sequence of $(M\!-\!1\!-\!a_2)M^{i_2}$ zeros needed to set all of the $( \mathfrak{i}_{\mathfrak{i}_2} \cdots \mathfrak{i}_0)$ indices to zero. 
 
 Continuing the count keeps repeating the second level pattern until $\mathfrak{i}_{i_3} = a_3$ so that there are $M^{i_3-i_2-1}$ repetitions of the second level pattern. A third level pattern is then completed by the sequence of $(M\!-\!1\!-\!a_3)M^{i_3}$ zeros needed to set all of the $(\mathfrak{i}_{i_3} \cdots \mathfrak{i}_0)$ indices to zero. Finally there is a $k$th level pattern with $M^{N\!-\! i_k \!-1\!}$ repetitions of the $(k\!-\!1)$th level pattern until the count is complete.  
\end{proof}

Equation (\ref{original_column}) can be expressed in the alternative form
  \begin{equation}
 \begin{array}{ll}
 \mathtt{X}^{i_k \cdots i_1}(a_k\cdots a_1) = \frac{1}{M^k} \: \huge{\cdot} \\ 
 \!\!\!\!\left[\begin{array}{ll}\underbrace{\:\cdots \:\underbrace{\overbrace{0}^{a_2M^{i_1}}\overbrace{\!0}^{M^{i_1}}\overbrace{\;\;\;\;\;0\;\;\;\;\;}^{(M\!-\!1\!-\!a_2)M^{i_1}}}_{\times a_2M^{i_2-i_1-1}}\underbrace{\overbrace{0}^{a_1M^{i_1}}\overbrace{1}^{M^{i_1}}\overbrace{\;\;\;\;\;0\;\;\;\;\;}^{(M\!-\!1\!-\!a_1)M^{i_1}}}_{\times M^{i_2-i_1-1}}\underbrace{\overbrace{0}^{a_2M^{i_1}}\overbrace{\!0}^{M^{i_1}}\overbrace{\;\;\;\;\;0\;\;\;\;\;}^{(M\!-\!1\!-\!a_2)M^{i_1}}}_{\times(M-1- a_2)M^{i_2-i_1-1}}\cdots}_{\times M^{N-i_k-1}}\end{array} \right]^\mathsf{T}\!\!\!\!. \label{column_form}
 \end{array}
 \end{equation}

Now consider the orthogonalisation process. With $\mathtt{X}^{i}$ denoting a column of the submatrix $X^{C_i}$, the first step obtains the projection of the column vector onto the orthogonal complement of the normalised constant vector $\overline{X}^0$ as 
\begin{equation}
\mathtt{X}^{i}_\perp = \mathtt{X}^{i} - (\overline{X}^0\cdot\mathtt{X}^{i})\overline{X}^0.   \label{orthog}
\end{equation}

Applying (\ref{orthog}) to the form (\ref{column_form}) gives
  \begin{equation}
 \begin{array}{ll}
 \mathtt{X}^{i_k \cdots i_1}_\perp(a_k\cdots a_1) = \frac{1}{M^k} \: \huge{\cdot} \\ 
 \left[\begin{array}{ll}\underbrace{\:\cdots \:\underbrace{\overbrace{-\!1}^{a_2M^{i_1}}\overbrace{\!-\!1}^{M^{i_1}}\overbrace{\;\;\;\;\;-\!1\;\;\;\;\;}^{(M\!-\!1\!-\!a_2)M^{i_1}}}_{\times a_2M^{i_2-i_1-1}}\underbrace{\overbrace{-\!1}^{a_1M^{i_1}}\overbrace{M^k\!\!-\!\!1}^{M^{i_1}}\overbrace{\;\;\;\;\;-\!1\;\;\;\;\;}^{(M\!-\!1\!-\!a_1)M^{i_1}}}_{\times M^{i_2-i_1-1}}\underbrace{\overbrace{-\!1}^{a_2M^{i_1}}\overbrace{\!-\!1}^{M^{i_1}}\overbrace{\;\;\;\;\;-\!1\;\;\;\;\;}^{(M\!-\!1\!-\!a_2)M^{i_1}}}_{\times(M-1- a_2)M^{i_2-i_1-1}}\cdots}_{\times M^{N-i_k-1}}\end{array} \right]^\mathsf{T}\!\!\!\!. \label{orthog_form}
 \end{array}
 \end{equation}

Note that because (\ref{orthog}) operates on rows only and the projection is onto the uniform vector, the hierarchical sequence of patterns in (\ref{column_form}) is preserved in (\ref{orthog_form}). 
 
Using the case $k=2$ with $a_1=a_2=0$ to illustrate the process, the orthogonalisation represented by (\ref{orthog}) to (\ref{orthog_form}) results in $\mathtt{X}^{i_2i_1}_\perp$. Then, the projection of this onto the orthogonal complement of the subspace formed by $\overline{\mathtt{X}}^{i_2}_\perp$ and $X^0$ is obtained by again applying (\ref{orthog}) resulting in $\mathtt{X}^{i_2i_1}_{\perp\perp}$. Finally, projecting this onto the orthogonal complement of the subspace formed by $\overline{\mathtt{X}}^{i_1}_\perp$, $\overline{\mathtt{X}}^{i_2}_\perp$ and $X^0$ yields
\begin{equation}
\begin{array}{ll}
\mathtt{X}^{i_2i_1}_{\perp\perp\perp}(0,0) = \frac{1}{M^2} \boldsymbol{\cdot} \left[\begin{array}{ll}\underbrace{\underbrace{\overbrace{(M\!\!-\!\!1)^2}^{M^{i_1}}\overbrace{-\!(M\!\!-\!\!1)}^{(M\!-\!1)M^{i_1}}}_{\times M^{i_2-i_1-1}}\underbrace{\overbrace{-\!(M\!\!-\!\!1)}^{M^{i_1}}\overbrace{1\cdots1}^{(M\!-\!1)M^{i_1}}}_{\times (M\!-\!1)M^{i_2-i_1-1}}}_{\times M^{(N-i_2-1)}}\end{array}\right]^\mathsf{T}\!\!. \label{orthog2}
\end{array}
\end{equation}
The form of (\ref{orthog2}) and numerical calculations suggest the $k$ attribute generalisation 
\begin{equation}
\begin{array}{ll}
\mathtt{X}^{i_k\cdots i_2 i_1}_{\perp\cdots\perp}(0\cdots 0) = \frac{1}{M^k}  \boldsymbol{\cdot} \\ \\
\!\!\!\!\!\!\!\!\left[
\begin{array}{ll}\underbrace{\underbrace{\overbrace{(M\!\!-\!\!1)^{k}}^{\times M^{i_1}}\;\;\overbrace{-\!(M\!\!-\!\!1)^{k-1}}^{\times (M\!-\!1)M^{i_1}}\;\;}_{\times M^{i_2-i_1-1}}\underbrace{\overbrace{-\!(M\!\!-\!\!1)^{k-1}}^{\times M^{i_1}}\;\;\overbrace{(M\!\!-\!\!1)^{k-2}}^{\times (M\!-\!1)M^{i_1}}}_{\times (M\!-\!1)M^{i_2-i_1-1}}}_{\times M^{(I_3-i_2-1)}} \\
\underbrace{\qquad\underbrace{\underbrace{\overbrace{-(M\!\!-\!\!1)^{k-1}}^{\times M^{i_1}}\;\;\overbrace{\!(M\!\!-\!\!1)^{k-2}}^{\times (M\!-\!1)M^{i_1}}\;\;}_{\times M^{i_2-i_1-1}}\underbrace{\overbrace{\!(M\!\!-\!\!1)^{k-2}}^{\times M^{i_1}}\;\;\overbrace{-(M\!\!-\!\!1)^{k-3}}^{\times (M\!-\!1)M^{i_1}}}_{\times (M\!-\!1)M^{i_2-i_1-1}}}_{\times (M\!-\!1) M^{(I_3-i_2-1)}}}_{\times M^{(i_4-i_3-1)}} \\
\qquad\qquad\underbrace{\underbrace{\overbrace{-(M\!\!-\!\!1)^{k-1}}^{\times M^{i_1}}\;\;\overbrace{\!(M\!\!-\!\!1)^{k-2}}^{\times (M\!-\!1)M^{i_1}}\;\;}_{\times M^{i_2-i_1-1}}\underbrace{\overbrace{\!(M\!\!-\!\!1)^{k-2}}^{\times M^{i_1}}\;\;\overbrace{-(M\!\!-\!\!1)^{k-3}}^{\times (M\!-\!1)M^{i_1}}}_{\times (M\!-\!1)M^{i_2-i_1-1}}}_{\times M^{(I_3-i_2-1)}} \\
\underbrace{\qquad\qquad\underbrace{\qquad\underbrace{\underbrace{\overbrace{(M\!\!-\!\!1)^{k-2}}^{\times M^{i_1}}\;\;\overbrace{\!-(M\!\!-\!\!1)^{k-3}}^{\times (M\!-\!1)M^{i_1}}\;\;}_{\times M^{i_2-i_1-1}}\underbrace{\overbrace{\!-(M\!\!-\!\!1)^{k-3}}^{\times M^{i_1}}\;\;\overbrace{(M\!\!-\!\!1)^{k-4}}^{\times (M\!-\!1)M^{i_1}}}_{\times (M\!-\!1)M^{i_2-i_1-1}}}_{\times (M\!-\!1) M^{(I_3-i_2-1)}}}_{\times (M-1)M^{(i_4-i_3-1)}}\;\;\mathbf{\cdots}}_{\times M^{N-i_k-1}}
\end{array}\right]^\mathsf{T}. \label{orthogk}
\end{array}
\end{equation}

\begin{lemma}
 Column vector $\mathtt{X}^{i_k\cdots i_2 i_1}_{\perp\cdots\perp}(0\cdots 0)$ of the submatrix ${\overline X}^{\boldsymbol{\mathsf{c}}_{k}^\ell}$  is given by (\ref{orthogk}).
 \label{lemma2}
\end{lemma}
\begin{proof}
Note that (\ref{orthogk}) is composed of basic units each of length $M^{i_1}$ each individual unit having identical components. These are grouped into a sequence of first level modules each containing $M^{i_1+1}$ components such that each of these modules sums to zero. In turn these first level modules are grouped into second level modules each containing $M^{i_2+1}$ components and so on. This hierarchical construction is continued eventually resulting in $k$th level modules each containing $M^{i_{k-1}+1}$ components. There are then $N\!-\!k\!-\!1$ of these $k$ level modules containing a total of $N$ components. Each module is composed of two submodules $\mathcal{S_M}_1$ and $\mathcal{S_M}_2$ such that 
\begin{equation}
\mathcal{S_M}_1 = -(M-1)\mathcal{S_M}_2.      \label{submod}
\end{equation}

 Orthogonality of (\ref{orthogk}) over the $\binom{N}{k}$ vectors $\mathbf{i}^{\mathbf{s}_\ell}_k $ can be seen by considering two values of $\ell$ say $\mathbf{i}^{\mathbf{s}_{\ell^\prime}}_k $and $\mathbf{i}^{\mathbf{s}_{\ell^{\prime\prime}}}_k $. If $i_1^\prime  <i_1^{\prime\prime}$ the length of the first level modules in vector $\mathtt{X}^{\prime\prime}$ is a multiple of the length of the first level modules in $\mathtt{X}^\prime$ so that the components of $\mathtt{X}^{\prime\prime}$ are uniform over every first level module of $\mathtt{X}^\prime$.  Consequently, in the scalar product of the two vectors, every first level module in the vector $\mathtt{X}^\prime$ sums to zero so that the vectors are orthogonal.

If $i_1^\prime=i_1^{\prime\prime}$, there must be some $j$ where $i_j^\prime\ne i_j^{\prime\prime}$. Suppose $i_j^\prime=i_j^{\prime\prime}-1$. For each $j-1$th module, in the scalar product of $\mathtt{X}^\prime$ and $\mathtt{X}^{\prime\prime}$, the submodule $\mathcal{S_M}_1^\prime$ is multiplied component by component by the submodule $\mathcal{S_M}_1^{\prime\prime}$. However each of the $M-1$ submodules $\mathcal{S_M}_2^\prime$ is also multiplied by $\mathcal{S_M}_1^{\prime\prime}$. From (\ref{submod}), 
\[\mathcal{S_M}_1^\prime \cdot \mathcal{S_M}_1^{\prime\prime} = -(M-1)\mathcal{S_M}_2^\prime \cdot \mathcal{S_M}_1^{\prime\prime}\]
but there are $M-1$ of the $\mathcal{S_M}_2^\prime \cdot \mathcal{S_M}_1^{\prime\prime}$ product components so that the scalar product contribution of each ($j-1$)th module is zero and the vectors are orthogonal. Clearly this is true also when $i_j^{\prime\prime}$ is any multiple of $i_j^\prime$. In fact this argument is independent of $k$ so that all of the vectors $\mathtt{X}^{i_k\cdots i_2 i_1}_{\perp\cdots\perp}(0\cdots 0)|_{k=1:N-1}$, including the main effects, are mutually orthogonal.

Because the orthogonalisation process of Section \ref{loglinear} begins with the leftmost column vectors of each submatrix, the remaining column vectors will be orthogonal to (\ref{orthogk}) so that (\ref{orthogk}) is the orthogonalised form of $ \mathtt{X}^{i_k \cdots i_1}(a_k\cdots a_1)$.
\end{proof}

For $\mathbf{a}_k = \langle 0,\cdots 1\rangle$ the column vector $\mathtt{X}^{i_k \cdots i_1}_{\perp \cdots \perp}(0,\cdots 1)$ lies in the orthogonal complement of  the uniform vector and $\mathtt{X}^{i_k \cdots i_1}_{\perp \cdots \perp}|(0, \cdots 0)$. Explicit derivation of the $k=2$ case and numerical calculations suggest that the $j$th second level module has the form
\begin{equation}
\begin{array}{ll}\underbrace{\underbrace{\overbrace{0}^{\times M^{i_1}}\;\;\overbrace{(M\!\!-\!\!1)^{k\!-\!j}(M\!\!-\!\!2)}^{\times M^{i_1}}\;\overbrace{\;-(M\!\!-\!\!1)^{k\!-\!j}\;\;}^{\times (M\!-\!2)M^{i_1}}}_{\times M^{i_2-i_1-1}}\underbrace{\overbrace{0}^{\times M^{i_1}}\overbrace{-(M\!\!-\!\!1)^{k\!-\!j\!-\!1}(M\!\!-\!\!2)}^{\times M^{i_1}}\;\overbrace{(M\!\!-\!\!1)^{k\!-j\!-\!1}}^{\times (M\!-\!2)M^{i_1}}}_{\times (M\!-\!1)M^{i_2-i_1-1}}}_{\times M^{(I_3-i_2-1)}} . \label{orthog01}
\end{array}
\end{equation}
 Note that this has the same modular structure as (\ref{orthogk}) and again the first level modules sum to zero. Indeed  the structures are sufficiently similar that the above orthogonality arguments can be invoked to show that (\ref{orthog01}) is orthogonal to (\ref{orthogk}) for all $k$ as well as to the corresponding $\mathtt{X}^{i_k^\prime\cdots i_2 i_1}_{\perp\cdots\perp}(0\cdots 1)$ for $k^\prime\ne k$. 
 
 Because the origin of this critical modularity is (\ref{orthog_form}), the orthogonalisation process applied to all of the remaining independent column vectors of $X^{\boldsymbol{\mathsf{c}}_{k}^\ell}$ will result in the orthogonalised submatrice ${\overline X}^{\boldsymbol{\mathsf{c}}_{k}^\ell}$ having the same modular structure as (\ref{orthogk}) and (\ref{orthog01}). 

\subsection{Conditional Subtables and their Geometric Means}
\label{Conditional}
Each term of the log linear expansion (\ref{xifunct}), $\xi_{\boldsymbol{\mathsf{c}_{k}^{\ell}}}$, involves the subset of 
$k$ attributes $\mathsf{c}_k^\ell$ which raises the question of how to handle the remaining $N\!-\!k$. While marginalisation is the common approach, the nonlinear $\log$ transformation makes the result difficult to interpret. Instead we adopt the alternative approach in which the members of $\mathsf{c}$ determine a subtable conditioned on the remaining attributes. 
 
 Recall from Section \ref{loglinear} that the $k$ attributes associated with the submatrix ${\overline X}^{{\boldsymbol{\mathsf{c}}_{k}^{\ell}}}$ are designated by $\mathbf{i}^{s_\ell}_k $ with the remaining $N\!-\!k$ attributes designated by the vector $\mathbf{i}^{g_\ell}_k$. Via (\ref{lex}), the column vectors of the submatrix have exactly the same indexing structure as the contingency table $\mathbf{T}$, so, from (\ref{column_index}), we can interpret an index vector of the form $\boldsymbol{\mathfrak{i}}$ as indexing a subtable designated by $\mathbf{g}^{k,\ell}$. This leads to:
 \begin{definition}
 Let the attribute vector $\mathbf{g}^{k,\ell}$ be fixed whereas the attribute vector $\mathbf{a}^{k,\ell}$ is allowed to range over the $M^k$ combinations of the attribute levels. Then the attribute index vector $\mathbf{i}^{g_{k,\ell}} = \langle i^g_{\ell_k} \cdots i^g_{\ell_1}\rangle$  specifies the conditioning attributes of the conditional subtable $\mathbf{T}_{\mathbf{a}^{k,\ell}}|_{\mathbf{g}^{k,\ell}}$ with indices given by, from (\ref{column_index})
  \begin{equation}
  \boldsymbol{\mathfrak{i}}^{s_{k,\ell}}_j = \mathbf{I}^{s_\ell}_ k\mathbf{a}^{k,\ell}+ \mathbf{I}^{g_\ell}_{N-k} \mathbf{g}^{k,\ell}_j   \label{subtable_index}
  \end{equation}
  where the subscript $j$ indicates that the index set is associated with the $j$th combination of conditioning attributes.
 \label{definitioncst}
 \end{definition}
 
 Consequently, it is the conditioning attributes which are associated with the bottom level modules in the modular hierarchy of (\ref{orthogk}) and (\ref{orthog01}) with (\ref{column_form}) ensuring that the elements in each of the bottom level modules are equal. Because of the row by row operation of the orthogonalisation process, the equality of the elements of these components is preserved by the orthogonalisation process even if the values of the elements end up differing from module to module. This is evident from (\ref{orthogk}) and (\ref{orthog01}). Furthermore each bottom level module represents one particular combination of attribute values.
 
We then have the following Lemma.
\begin{lemma}
\label{Lemma_main}
Let $\overline{\mathfrak{X}}$ be the orthogonalised $M^k \! \times \!M^k$ design matrix and let $\boldsymbol{\Gamma}^\ell_k$ be the table of log transformed geometric means, over the conditioning attributes $\mathbf{g}^{k,\ell}$, of corresponding entries in the $M^{N-k}$ conditional subtables $\mathbf{T}\!_{\mathbf{a}^{k,\ell}}|_{\mathbf{g}^{k,\ell}}$. Then,
 the magnitude of the projection of  $\boldsymbol{\Gamma}^\ell_k$ onto the subspace $\widehat{\mathfrak{X}}^{\boldsymbol{\mathsf{c}}^\ell_k}_\perp$, $\mathbf{\mathcal{X}}^\perp_{k\ell}$, is given by
 \begin{equation}
 |\mathbf{\mathcal{X}}^\perp_{k\ell}| = M^{-\frac{N-k}{2}}|\boldsymbol{\chi}^\perp_{k\ell}|   \label{gmproj}
 \end{equation}
 where, from (\ref{magsq}), $|\boldsymbol{\chi}^\perp_{k\ell}|$ is the magnitude of the projection of the originating table vector $\mathbf{T}$ onto the subspace $\widehat{X}^{\boldsymbol{\mathsf{c}}^{\ell}_{k}}_{\perp}$.
 
\end{lemma}
\begin{proof}
It is clear from (\ref{column_index}) and the construction of the basis vectors that the $1$'s in (\ref{original_column}) constitute an indicator vector for the conditioning indices 
 \begin{equation}
 \boldsymbol{\mathfrak{i}}^{g_{k,\ell}}_j =  \mathbf{I}^{s_\ell}_ k\mathbf{a}^{k,\ell}_j + \mathbf{I}^{g_\ell}_{N-k}\mathbf{g}^{k,\ell} \label{cond_indices}
 \end{equation}
 associated with the $j$th combination of conditional attribute values as $\mathbf{g}^{k,\ell}$ ranges over the $M^{N-k}$ combinations of the conditioning attributes.
The scalar product of the $\nu$th column of ${\overline X}^{{\boldsymbol{\mathsf{c}}_{k}^{\ell}}}$ and the table vector $\mathbf{T}$ is composed of the sums of the $M^k$ individual scalar products of the subtables $\mathbf{T}(\boldsymbol{\mathfrak{i}}^{g_{k,\ell}}_j)$ and the corresponding $\overline{X}^{\boldsymbol{\mathsf{c}}^\ell_k}_\nu (\boldsymbol{\mathfrak{i}}^{g_{k,\ell}}_j)$, $j\in \{1\cdots M^k\}$. However, because all of the elements of the latter are equal, each of the individual scalar products is simply the sum of the values in the particular subtable $\mathbf{T}(\boldsymbol{\mathfrak{i}}^{g_{k,\ell}}_j)$ multiplied by one of the elements of $\overline{X}^{\boldsymbol{\mathsf{c}}^\ell_k}_\nu (\boldsymbol{\mathfrak{i}}^{g_{k,\ell}}_j)$.

 The $M^k$ values of $\mathbf{a}^{k,\ell}$ then generate the elements of an $M^k$ dimensional column vector $\overline{\mathfrak{X}}^{\boldsymbol{\mathsf{c}}^\ell_k}_\nu$ by sampling the column vector $\overline{X}^{\boldsymbol{\mathsf{c}}^\ell_k}_\nu$  at the entry indexed by (\ref{cond_indices}) for every $\mathbf{a}^{k,\ell}_j$. This structure is shared by all $(M\!-\!1)^k$ columns of the submatrix $\overline{X}^{\boldsymbol{\mathsf{c}}^\ell_k}$ so the resultant vectors can be collected into the  submatrix $\overline{\mathfrak{X}}^{\boldsymbol{\mathsf{c}}^\ell_k}$.  In fact it is clear from the proof of Lemma \ref{lemma1} and (\ref{column_form}) that the sampling is performed by eliminating the repetition as the row index count proceeds from $i_j$ to $i_{j+1}$ which requires the indices $i_1=0$ and $I_{j+1}=I_j+1$ in the reduced matrix ensuring $i_k=k-1$. 

With $\mathfrak{D}_{\boldsymbol{\mathsf{c}}_{k}^{\ell}}=\overline{\mathfrak{X}}^{{\boldsymbol{\mathsf{c}}_{k}^{\ell}}^\mathsf{T}}\overline{\mathfrak{X}}^{\boldsymbol{\mathsf{c}}_{k}^{\ell}}$, this enables the magnitude of the function  $\xi_{\boldsymbol{\mathsf{c_{2}^{\ell}}}}$ to be expressed as, using (\ref{magsq}), 
\begin{equation}
 |\boldsymbol{\chi}^\perp_{k\ell}| = M^{-\frac{(N-k)}{2}} \!\left|\mathfrak{D}^{-\frac{1}{2}}_{\boldsymbol{\mathsf{c}}_{k}^{\ell}}{\overline{\mathfrak{X}}}^{{\boldsymbol{\mathsf{c}}_{k}^\ell}^{\mathsf{T}}}\sum_{\mathbf{g}^{k,\ell}}\mathbf{T}\!_{\mathbf{a}^{k,\ell}}|_{\mathbf{g}^{k,\ell}} \right| = M^{\frac{(N-k)}{2}} \left|\mathfrak{D}^{-\frac{1}{2}}_{\boldsymbol{\mathsf{c}}_{k}^{\ell}}{\overline{\mathfrak{X}}^{\boldsymbol{\mathsf{c}}_{k}^\ell}}^{\mathsf{T}}\frac{\sum_{\mathbf{g}^{k,\ell}}\mathbf{T}\!_{\mathbf{a}^{k,\ell}}|_{\mathbf{g}^{k,\ell}}}{M^{(N-k)}} \right|.   \label{condmag2}
 \end{equation}

Inverting the $\log$ transform, the summation term in (\ref{condmag2}) becomes
\begin{equation}
\mathcal{T}_k^\ell = \left(\prod_{\mathbf{g}^{k,\ell}}\mathcal{T}_{\mathbf{a}^{k,\ell}}|_{\mathbf{g}^{k,\ell}}\right)^{\frac{1}{M^{(N-2)}}}   \label{geo_mean} 
\end{equation}
which is the geometric mean of the conditional subtable $\mathcal{T}_{\mathbf{a}^{k,\ell}}|_{\mathbf{g}^{k,\ell}}$ over the conditioning attributes $\mathbf{g}^{k,\ell}$.
\end{proof}

This leads to our main theoretical result:
\begin{theorem}
\label{Theorem1}
Let $\overline{\mathfrak{X}}_0$ be the $M^{k_0}\times M^{k_0}$ design matrix and $\boldsymbol{\Gamma}^\ell_{k_0}$ be the log geometric mean of the conditional subtables $\mathbf{T}\!_{\mathbf{a}^{k_0,\ell}}|_{\mathbf{g}^{k_0,\ell}}$ over the conditioning attributes $\mathbf{g}^{k_0,\ell}$. Then the magnitude of the projection of $\boldsymbol{\Gamma}^\ell_{k_0}$ onto the subspace ${\widehat{\mathfrak{X}}^{\boldsymbol{\mathsf{c}}_{k}^\ell}_\perp}$ for $k\le k_0$ is given by
\begin{equation}
 |\boldsymbol{\chi}^\perp_{k\ell}|  = M^{\frac{(N-k_0)}{2}} \left|\mathfrak{D}^{-\frac{1}{2}}_{\boldsymbol{\mathsf{c}}_{k}^{\ell}}{\overline{\mathfrak{X}}^{\boldsymbol{\mathsf{c}}_{k}^\ell}}^{\mathsf{T}} \boldsymbol{\Gamma}^\ell_{k_0}\right|.   \label{part_gm_proj}
 \end{equation} 
Furthermore the magnitude of the projection of the original table vector onto the $M^{k_0}$ dimensional subspace $\widehat{X}^{\boldsymbol{\mathsf{c}}^{\ell}_{k}}$, $\boldsymbol{\chi}_{k_0\ell}$, orthogonal to the uniform vector is given by 
\begin{equation}
|\boldsymbol{\chi}_{k_0\ell}| = M^\frac{(N-k_0)}{2}\left|\bf{\mathfrak{T}}_{k_0\ell}\right|  \label{gm_proj}
\end{equation}
where $\bf{\mathfrak{T}}_{k_0\ell}$ is the magnitude of the projection of $\Gamma^\ell_{k_0}$ onto the $M^{k_0}$ dimensional subspace ${\widehat{\mathfrak{X}}^{\boldsymbol{\mathsf{c}}_{k}^\ell}}$ , orthogonal to the uniform vector. 
\end{theorem}
\begin{proof}
From (\ref{original_column}), the number of $1$s in a pre-orthogonalised column of ${\overline X}^{\boldsymbol{\mathsf{c}}_{k_0}^{\ell}}$ is $M^{N\!-\!k_0}$ which is the number of combinations of the conditioning attributes $\mathbf{g}^{k_0,\ell}$. For $k=k_0\!-\!j$, $j$ attributes have been changed from conditional to conditioning so  there are $k$ sets of $k_0\!-\!j$ attributes and, for each set, $N\!-\!k_0\!+\!j$ conditioning attributes.  Because of the additional conditioning attribute, the set of $M^{N\!-\!k_0}$ $1$s is expanded by a factor of $M^j$ in a pattern determined by the form of (\ref{original_column}). These $1$s are transformed to equal elements in the orthogonalisation process. 

As in the proof of Lemma \ref{Lemma_main}, the scalar product of $\boldsymbol{\Gamma}^\ell_{k_0}$ and a column of the submatrix ${\overline{\mathfrak{X}}^{\boldsymbol{\mathsf{c}}_{k}^\ell}_k}$ is composed of sums of $M^j$ terms from $\boldsymbol{\Gamma}^\ell_{k_0}$ each sum multiplied by an identical element of the submatrix column. Then
\[ {\overline{\mathfrak{X}}^{\boldsymbol{\mathsf{c}}_{k}^\ell}}^{\mathsf{T}}\boldsymbol{\Gamma}^\ell_{k} = {\overline{\mathfrak{X}}^{\boldsymbol{\mathsf{c}}_{k_0}^\ell}}^{\mathsf{T}}\boldsymbol{\Gamma}^\ell_{k_0}/M^j. \]
With $\mathfrak{D}^{-\frac{1}{2}}_{\boldsymbol{\mathsf{c}}_{k}^{\ell}} = \mathfrak{D}^{-\frac{1}{2}}_{\boldsymbol{\mathsf{c}}_{k_0}^{\ell}}M^{\frac{j}{2}}$ this gives (\ref{part_gm_proj}). 
Equation (\ref{gm_proj}) then follows because of the orthogonality of the subspaces means that $|\boldsymbol{\chi}_{k_0\ell}|$ and $\left|\bf{\mathfrak{T}}_{k_0\ell}\right|$ are each given by the square root of the sum of the squares of the magnitudes of the respective projections onto the individual subspaces.
\end{proof}

The immediate implication of Theorem \ref{Theorem1} is that `marginalisation' of a contingency table to reduce the number of attributes using a geometric mean preserves the interaction structure of the table up to the interactions between all of the conditional attributes. 

There are also implications for the collapsibility \cite{Dar1} of the table although investigation of this aspect is beyond the scope of the current paper. Further implications for cyber security are considered in the Discussion, Section \ref{discussion}, below.

\subsection{Probabilistic Salience}
\label{salience}

That the projection of a log table vector onto the subspace ${\widehat{\mathfrak{X}}^{\boldsymbol{\mathsf{c}}_{k}^\ell}}$ is orthogonal to the corresponding uniform vector suggests that the larger the magnitude of this projection is in proportion to the vector magnitude, the more the table entries are concentrated on a small number of cells.  That is to say, the more salient are those cells. This motivates the following definition.

\begin{definition}
\label{salience_factor}
Let $C_s\subset C$ be a subset of the set of attributes $C$ and $\mathcal{T}_s$ be a conditional subtable in the attributes $C_s$.
\begin{description}
\item[a)] If $|\chi_s|$ is the magnitude of the projection of the log linear transform, $T_s$,  of $\mathcal{T}_s$ onto the subspace orthogonal to the uniform subtable vector $T_{su}$, then the {\bf Probabilistic Salience}  of the conditional subtable $\mathcal{T}_s$ is
\begin{equation}
\psi_s = \frac{|\chi_s|}{|T_s|}.                     \label{psidef}
\end{equation}
\item[b)] Let $\mathcal{T}_{gm}$ be the geometric mean of all of the conditional subtables in $C_s$ generated by assigning values to the conditioning attributes $C\!\setminus\! C_s$ and $T_{gm}$ be its log linear transform. If $|\chi_{gm}|$ is the magnitude of the projection of $T_{gm}$ onto the subspace orthogonal to the uniform subtable vector $T_{gmu}$ then the {\bf Probabilistic Salience Function} of the attribute subset $C_s$ is
\begin{equation}
\Psi_{gm} = \frac{|\chi_{gm}|}{|T_{gm}|}            \label{Psidef}
\end{equation}
\end{description}
\end{definition}

Our proposal is that Probabilistic Salience is a measure of the degree to which a contingency table is dominated by or concentrated on a small number of attribute values, a proposal which we will vindicate in section \ref{vindication} below.

\subsection{Potential Application  - Interaction Limiting}
\label{limiting}
While the organisational counter measures discussed in the Introduction provide one line of defence against social engineering attacks, we propose an alternative approach which employs the techniques developed in this paper in a different role to that which has been the focus thus far.  It has been demonstrated multiple times that releases of statistics from de-identified databases can be processed to re-identify individuals. Differential privacy techniques introduce controlled noise into the data such that the resultant statistics have a bounded error but the risk of re-identification is significantly reduced. Application of these techniques to contingency table data is described in \cite{Yang2} using a very different form of log linear model to that employed here.

By analogy, the objective of what we refer to as `de-personalisation' is to inhibit the inference of attribute values in vulnerable cases while retaining the overall character of the data.  De-personalisation is performed by firstly  calculating the log linear expansion coefficients $\boldsymbol{\beta}$ using (\ref{betaV}), then setting $\boldsymbol{\beta}_{k\ell}=0$ when $k>k_\dag$ for some $k_\dag$. Finally the modified data is reconstructed using (\ref{orthexp}). This limits the interaction to that between $k_\dag$ attributes, in a sense, blurring the cell contents. We refer to this as interaction limiting the data by analogy with bandlimiting in Fourier transform based signal processing.

Suppose $k_\dag<k_0$ and suppose the marginalised data is publicly released in the form of $k_0$ dimensional geometrical mean subtables $\mathcal{T}^\ell_{k_0}$ involving subsets of $k_0$ attributes. Theorem \ref{Theorem1} ensures that the maximum interaction order of the released data is limited to $k_\dag$.  It then follows from definition \ref{salience_factor} {\bf b} via (\ref{Psidef}) that the Probabilistic Salience Function of the released data is less than it would have been had the data not been de-personalised.

However simple interaction limiting as described is a heavy handed method of de-personalisation. A more subtle approach would be to use the Probabilistic Salience Function to detect subsets of attributes with high interactivity via their geometric means as demonstrated in the top left panel of Figure \ref{case27}. The interaction structure of high interactivity subsets  could then be investigated using the reduced design matrix $\mathfrak{X}$ with Theorem \ref{Theorem1} providing the link back to the full table. High value interaction coefficients $\boldsymbol{\beta}$ (Section \ref{xi}) can be selectively set to zero although some care is required to maintain the hierarchical nature of the expansion \cite{DLS}\cite{Dar1}. 

While this is necessarily a brief sketch of the proposal we believe that the theory described here provides a solid foundation for a detailed investigation and development of de-personalisation techniques. However this lies well beyond the scope of this paper.

\subsection{Probabilistic Salience as a Measure of Concentration}
\label{vindication}

We make the basic assumption that the contingency table is a member of an ensemble of similar tables with the common characteristic that they represent the same population so that their respective entries sum to the population size say $N_T$. More formally this is equivalent to assuming that each table is drawn from a multinomial distribution although we will not make use of this fact. Each table vector $\mathcal{T}^0$  is then a point on a simplex, $\mathcal{S}^0$ \cite{Bronst}, and can be expressed in the form

\begin{equation}
\mathcal{T}^0 = \sum_{k=1}^{M_T} \lambda_k V^0_k  \;\;\ : \;\; \sum_{k=1}^{M_T} \lambda_k = 1 \nonumber
\end{equation}
where $\lambda_k, k\in \{1\;M_T\}$ are the components of a parameter vector $\Lambda$ and the $M_T$ dimensional vertex vector $V^0_k = [\underbrace{0 \cdots 0\; 0}_{k-1}\;N_T\; 0\cdots 0]^\mathsf{T}$
so that $\sum_{\ell=1}^{M_T} \tau^0_\ell = N_T.$

However, to avoid having to deal with $-\infty$ issues caused by zero entries in the table, we will work with the vector $\mathbf{\mathcal{T}}$ derived from the affine mapping 
\[ V=\frac{N_T-M_T}{N_T} V^0 +1\]
 of $\mathcal{S}^0$ which results in a new simplex $\mathcal{S}$. This is the convex hull of its vertices $V_k\;;\;k=1:M_T$ where
\[ V_k = [\underbrace{1 \cdots 1\;\; 1}_{k-1}\;(N_T-M_T+1)\; \;1\cdots 1]^\mathsf{T} \]
and
\begin{equation}
\mathcal{T}  =  \left[\lambda_1(N_T-M_T)+1\;\cdots\;\lambda_k(N_T-M_T)+1\;\cdots\;\lambda_{M_T}(N_T-M_T)+1\right]^\mathsf{T}.  \label{Tmap}
\end{equation}
It is easily verified that again $\sum_{\ell=1}^{M_T} \tau_\ell = N_T.$

Then the log transformed vector $\mathbf{T}$ has the form
\begin{equation}
\mathbf{T} = [\;\cdots\;\; \log(\lambda_k(N_T-M_T)+1)\;\; \cdots\;]^\mathsf{T}    \label{Tterm}
\end{equation}
with $0\le t_k \le \log(N_T-M_T)$.

The above proposal can be validated by finding a vector $\mathbf{T}$ such that the magnitude of the projection of $\mathbf{T}$ onto  the subspace orthogonal to the uniform vector is maximised subject to the constraints that the components of $\mathcal{T}$ sum to $N_T$ and the components of (\ref{Tterm}) are positive. However, rather than solve this constrained optimisation problem directly we will tackle an approximation as defined below - an approximation which improves as $N_T$ increases.

The ensemble of vectors (\ref{Tterm}) can be represented in functional form by expressing, without loss of generality, the component $t_1$ as a function of the remaining components with the $\lambda$s as parameters. From (\ref{Tmap}),
\begin{equation}
 \lambda_k = \frac{e^{t_k} -1}{(N_T-M_T)}          \label{lambdak}
 \end{equation}
which, with
\[ \tau_1 = \lambda_1(N_T-M_T)+\sum_{k=1}^{M_T}\lambda_k = \left(1-\sum_{k=2}^{M_T} \lambda_k\right)\left(N_T-M_T\right)+1\]
leads to
\begin{equation}
t_1 = \log\left(N_T - \sum_{k=2}^{M_T} e^{t_k}\right)\quad:\quad\sum_{k=2}^{M_T} e^{t_k}\le N_T.  \label{T1}
\end{equation}

Now the exponential function is convex and the sum of exponential functions is convex so the negative of the sum of exponential functions is concave \cite{B&V}. Adding $N_T$ to the argument of the $\log$ function retains its concavity so, because the $\log$ function is itself concave, (\ref{T1}) is a concave function. The significance of this here is that the function (\ref{T1}) is topographically simple with no inconvenient local maxima and minima, no hidden valleys etc.

We can understand the `shape' of this function by considering the transform of the vertices of the simplex $\mathcal{S}$ which are described by vectors of the $\lambda$ parameters having the form
\[ [0 \cdots \; 0\;1\;0\;\cdots\;0]^\mathsf{T}. \]
Then, from (\ref{Tterm}), the transformed simplex vertices have the form, 
\[ T_v = [0\cdots\;0\;\log(N_T-M_T+1)\;0\;\cdots\;0]^\mathsf{T} \] suggesting that they be viewed as those vertices of a hypercube which are adjacent to the origin, with the hypercube having edges of length $\log(N_T-M_T+1)$. Each vertex of the hypercube is described by a vector $T^h_v$ which has some number, $r$, of components equal to $\log(N_T-M_T+1)$ with the remainder zero. 

Furthermore if we consider the full set of parameter vectors, $\Lambda_r$, having $r$ nonzero components all equal to $1/r$ for $r=1:M_T-1$, we can see from (\ref{Tterm}) that the corresponding vector $T^{pv}$ has the same number of zeros and that the nonzero components are 
\begin{equation}
\log\left((N_T-M_T)/r+1\right)\le \log(N_T-M_T+1).   \label{Svertices}
\end{equation}
We can then associate each of these vectors $T^{pv}$ with that vertex of the hypercube, $T^h$, which has corresponding nonzero components. 

Again without loss of generality, select a vertex other than the origin, say $V^h_r$, and let the vector $V^h_r$ have $r$ non zero components. Now select any two of those components, say the $\ell$th and $j$th components thereby defining a coordinate plane containing the $\ell$th and $j$th axes of the hypercube. Consider the curve formed by the intersection of the function (\ref{T1}) with a plane parallel to the coordinate plane defined by the $\ell$th and $j$th axes and through the point $T^{pv}_r$ in which the $\ell$th and $j$th components are $t_\ell$ and $t_j$. We can assume that $\ell\ne1$ and $j\ne1$. If one of them is $=1$ or if $t_1=0$, then the following argument is simplified slightly but comes to the same result. 

Because at a vertex, $\lambda_k$ is either $1/r$ or $0$, from (\ref{lambdak}) $e^{t_k}$ is either $\frac{N_T-M_T}{r}+1$, $e^{t_\ell}$, $e^{t_j}$ or $e^0$. The curve is then described, from (\ref{Tterm}) and (\ref{T1}), by 
\begin{equation}
\begin{array}{l}
 \log\left(\frac{(N_T-M_T)}{r} + 1\right) =     \nonumber \\
\hspace{15mm}   \log\left[(N_T-\left((r-3)(\frac{N_T-M_T}{r}+1) +e^{t_\ell} + e^{t_j}+ (M_T-r)e^0\right)\right] \nonumber
 \end{array}
 \end{equation}
 which leads to
 \begin{equation}
 t_\ell = \log\left[2\left(\frac{N_T-M_T}{r}+1\right)-e^{t_j}\right]=\log\left(A-e^{t_j}\right).\label{curve}
\end{equation}

The standard formula for the curvature $\kappa$ of $y=f(x)$ is
\[ \kappa = \frac{|y''|}{\left(1+(y')^2\right)^{\frac{3}{2}}}. \]
Differentiating, assuming $y''\ne0$, gives,
\begin{equation}
 \frac{d\kappa}{dx} = \frac{1}{\left(1+(y')^2\right)^{\frac{3}{2}}} \left(y'''  - \frac{3y'y''^2}{1+(y')^2}\right)  \label{kd}
\end{equation}
 Then, with $t_\ell=y$ and $t_j = x$ in (\ref{curve})
 \[ y' = \frac{-e^x}{A-e^x}, \;\;\;\;\;y'' = -\left(\frac{e^x}{A-e^x} + \frac{e^{2x}}{(A-e^x)^2}\right) \]
 and
 \[ y''' = -\left(\frac{e^x}{A-e^x} + 3\frac{e^{2x}}{(A-e^x)^2} + 2\frac{e^{3x}}{(A-e^x)^3} \right). \]
 
At the point $T^{pv}_r$ (from (\ref{Tterm}))
\[ x = y = t_j = \log\left(\frac{N_T-M_T}{r}+1\right) = \log \frac{A}{2} \]
which results in $y' = -1$, $y'' = -2$, and $y''' = -6$ so that the differential of the curvature is, from (\ref{kd})
\begin{equation}
 \frac{d\kappa}{dx} = 0.    \label{kmax}
 \end{equation}

At $T^{pv}_r$, $\kappa = 2^{-\frac{1}{2}}$. Setting $e^{t_j} = \lambda_j(N_T-M_T)+1$ from (\ref{Tterm}) it is evident that as $\lambda_j$, say, decreases from $1/r$, the magnitude of the numerator of $y'$ decreases whereas the denominator increases so that $|y'|$ decreases from unity.
Now $y'' = y'(1-y')$ and the curvature can be expressed as 
\[\kappa = \frac{y'(1-y')}{(1+y'^2)^{\frac{3}{2}}}. \]
For $|y'|<<1$, $|\kappa|\approx|y'|<1/\sqrt{2}$ which, with (\ref{kmax}) demonstrates that the point $T^{pv}_r$ is the point of maximum curvature of the curve (\ref{curve}). 

This argument holds for any of the $\binom{r}{2}$ curves defined by a pair of non-zero components from the vector $\Lambda_r$, these curves representing mutually orthogonal two dimensional cross sections of the polytope through the vertex $T^{pv}_r$.  The function (\ref{T1}) in the vicinity of $T^{pv}_r$ can therefore be visualised as a rounded version of the vertex $T^h_r$ into which it fits. Note that the curvature is independent of the scaling factors of the function, $N_T$ and $M_T$. Consequently as the size of the hypercube increases with $N_T$ the curvature remains constant as the function (\ref{T1}) expands implying that the function fits deeper into the corner formed by the vertex i.e. that $T^h_r$ becomes a better approximation to $T^{pv}_r$.  We will refer to $T^{pv}_r$ as pseudo vertices. 

At the transformed vertices of the simplex, $r=1$ and can we see from (\ref{Svertices}) that the function coincides with the hypercube at those vertices, i.e. at $T^h_1$. All of this leads to the conclusion that the function (\ref{T1}) is well approximated by the relative boundary of the hypercube excluding those faces containing the origin.

However, as described in Section \ref{loglinear}, our measure of probabilistic structure is derived from the projection of a $\log$ transformed contingency table vector onto the orthogonal subspace, $\widehat{X}^{C_{0}}_\perp$,  of the constant vector $\widehat{X}^{0}$. We will now approximate that projection by first approximating $T$ with a point on the hypercube and then projecting that point onto $\widehat{X}^{C_{0}}_\perp$ noting that the projection of the entire hypercube is a convex polytope $P$. For example in three dimensions, the projection of a cube onto $\widehat{X}^{C_{0}}_\perp$ is a regular hexagon. To simplify the notation we will work with the unit hypercube having vertices $T^{\bar{h}}_r$ with the understanding that the final results are to be scaled by $\log(N_T-M_T+1)$. 

The projection of the vector $T^{\bar{h}}_r$ onto $\widehat{X}^{0}$ is of length $r/\sqrt{M_T}$ and subtracting the projection $\hat{\mathbf{1}}\cdot r/M_T$ from $T^h_r$ results in  $T^p_r$, the projection onto $\widehat{X}^{C_{0}}_\perp$. This has $r$ components equal to $1-r/M_T$ and $M_T-r$ components equal to $-r/M_T$. Its magnitude squared is
\begin{equation}
\mathsf{L}^2(T^p_r) = r(1-r/M_T)^2+(M_T-r)r^2/M^2_T= r(1-r/M_T). \label{L2}
\end{equation}
The squared magnitude of $T^{h1}_r$ is $r$ so that the Probabilistic Salience is, from (\ref{psidef}),
\begin{equation}
\psi = \sqrt{1-r/M_T}    \label{psi}
\end{equation}

Now consider any two vertices of the hypercube other than the all ones vertex and the all zeros vertex. Let these be $T^{h0}_r$ and $T^{h1}_s$ and their projections $T^{p0}_r$ and $T^{p1}_s$. It is clear that there must be some component, say the $\ell$th, of $T^{p0}_r$ which is $1-r/{M_T}$ whereas the $\ell$th component of $T^{p1}_{s}$ is $-r/{M_T}$, so that the $\ell$th term in the scalar product of $T^{p0}_r$ and $T^{p1}_s$ is negative whereas the corresponding term in ${T_r^{p0}}^\mathsf{T}\cdot T^{p0}_r$ is positive. Consequently 
\begin{equation}
{T_r^{p0}}^\mathsf{T}\cdot T^{p1}_{s} <  \mathsf{L}^2(T^{p0}_r)    \label{SPinequal}
\end{equation}
The line $0 - T^{p0}_r$  is normal to a hyperplane $H^0$  such that $T^{p0}_r \in H^0$ and (\ref{SPinequal}) shows that this hyperplane bounds a closed halfspace $k^0$ which contains all of the other vertex projections.

An m-dimensional  hypercube is constructed from hypercubes of lower dimension so its smallest faces are points (vertices), lines (edges) and squares. Select two of the hypercube vertices which are adjacent to some vertex $T^{h0}_r|0<r<M_T$ by taking the complements of, say, the $\ell$th and $j$th components respectively. Assume for the sake of being definite that the $\ell$th component of $T^{h0}_r$ is a one and the $j$th component is a zero so that taking the complement of the $\ell$th component gives the vertex $T^{h1}_{r-1}$ and the complement of the $j$th component gives the vertex  $T^{h2}_{r+1}$. Then define a fourth vertex adjacent to both $T^{h1}_{r-1}$ and $T^{h2}_{r+1}$ by complementing the $j$th  component of $T^{h1}_{r-1}$  and the $\ell$th component of $T^{h2}_{r+1}$ to give the $T^{h3}_r$.

These four vertices define a 2-face of the hypercube which is a two dimensional polytope, in fact a square, and, because the hypercube vertices are the extreme points of the hypercube, these four vertices are the extreme points of the face (\cite{Bronst}, Theorem 7.3). The face is then the convex hull of these four vertices (\cite{Bronst},Theorem 7.2) i.e. the set of all convex combinations of these vertices taken three at a time (\cite{Bronst} Corollary 5.11) . 

Now let $T^{p0}_r$, $T^{p1}_{r-1}$, $T^{p2}_{r+1}$, and $T^{p3}_r$ be the projections of $T^{h0}_r$, $T^{h1}_{r-1}$, $T^{h2}_{r+1}$ and $T^{h3}_r$. Define $T^{pc}_{r-1}$ as the midpoint of the line segment $\{T^{p0}_r,T^{p3}_r\}$. This is also the midpoint of the line segment $\{T^{p1}_{r-1},T^{p2}_{r+1}\}$ implying that the points $T^{p0}_r$, $T^{p1}_{r-1}$, $T^{p2}_{r+1}$, and $T^{p3}_r$ all lie on a plane which is, in fact, the projection of the corresponding hypercube face as expected from the linearity of the projection operation. Furthermore, because the projection is a linear operation, any point, $T^v$, on the projection of the face in question is a convex combination of $T^{p0}_r$, $T^{p1}_{r-1}$, $T^{p2}_{r+1}$, and $T^{p3}_r$ i.e, for $\lambda_1,\lambda_2,\lambda_3, \lambda_4 \ge 0$ and $\lambda_1+ \lambda_2 + \lambda_3 + \lambda_4=1,$
 \[T^v = \sum_{k=1}^4 \lambda_k T^{pk}_\centerdot.\]
 Then
\begin{equation}
{T^{p0}_r}^\mathsf{T} \cdot T^v = \sum_{k=1}^4 \lambda_k {T^{p0}_r}^\mathsf{T}\cdot T^{pk}_\centerdot < \mathsf{L}^2(T^{p0}_r) = \mathsf{L}^2(T^{p3}_r).
 \label{faceinequal}
\end{equation}

Consequently the whole projected 2-face is contained in the halfspace $k$. This argument can be repeated for all of the other 2-faces of the hypercube of which $T^{h0}_r$ is a member establishing that $k$ is a supporting halfspace and $H^0$ is a supporting hyperplane of $P$. Furthermore $H^0\cap P = \{T^{p0}_r\}$ demonstrating that $T^{p0}_r$ is an extreme point of $P$ i.e. a vertex. So we can conclude that the projections of the vertices of the hypercube excluding the all ones and all zeros vertices are in fact the vertices of the polytope $P$.

In the context of an optimisation problem, $|T^v|^2$ is a convex function with the projected 2-face acting as the feasible set and, because there are no distinct local minima or maxima, it attains its global minimum and maximum on the relative boundary of the projected 2-face (\cite{B&V}). Indeed, it follows from \cite{B&V} equation (4.21) that the maximum must occur at some vertex of the projected 2-face and the minimum at another.

Let $T^{v\nu}$ be a point on projected 2-face $\nu$ and $T^{v\mu}$ be a point on a distinct projected 2-face $\mu$. Furthermore let the maximum length vertices be the projections of hypercube vertices with $r^+_\nu$ and $r^+_\mu$ respectively and the minimum length vertices be the projections of hypercube vertices with $r^-_\nu$ and $r^-_\mu$ respectively. Then if $r^+_\nu<r^-_\mu$, from (\ref{psi}), $\psi_\nu>\psi_\mu$. This demonstrates that Probabilistic Salience is a strong indication of the degree to which the table entries are concentrated on a few cells vindicating our proposal above.

\section{Computational Results}
\label{results}

To investigate the type of result that the process can produce, Australian Census data was obtained from the Australian Bureau of Statistics (ABS) using their DATALAB facilities. Security restrictions imposed by the ABS  limited us to 9 attributes and then only with a high level of aggregation, the aggregation scheme being designed by us.  We limited the analysis to seven attributes resulting in 36 distinct cases i.e. nine attributes taken seven at a time. Comparison of the results of applying our algorithm to these 36 cases gave us a good sense of how it responded to changes in the attribute combination. One of the better cases is presented in Fig. \ref{case27} where we focus on bivariate conditional subtables.

While we acknowledge that this is a very small data set by data processing standards, we stress that the purpose is to provide a simple illustration of the operation of our process to enhance understanding rather than to explore large scale structure. The data set is bland in the sense that each attribute interacts strongly with every other attribute but we believe that this bland data set provides a sensitive test of the model's ability to discriminate between similar levels of structure.

The top left hand panel in each figure provides a visual bar graph comparison of the Probabilistic Salience Function values, $\Psi$, for the geometric mean conditional subtables of each pair of attributes. In spite of the blandness of the data and the limitation to pairwise interactions, the bar graph indicates that Probabilistic Salience is able to clearly differentiate between subsets.

The top right panel displays, for the minimum and maximum values of $\Psi$, histograms of the Probabilistic Salience $\psi$ of individual conditional subtables over the set of conditioning attributes. In each case the $\psi$ values are reasonably well clustered around the $\Psi$ values of the respective geometric mean although there is significant dispersion toward higher values in each case. The reason for this becomes is revealed by the remaining panels which show three dimensional histograms of the bivariate conditional subtable entries for the conditioning attribute values listed. For the minimum and maximum $\Psi$ values these are for the subtable closest to the geometric mean (middle) and for the maximum $\psi$ (bottom).

Numerical checks confirm the visual impression that the maximum $\psi$ values are generated by subtables with (originally) many zero entries - indeed the bottom left subtable with the highest $\psi$ has only a single non zero entry. Although this is, somewhat ironically, associated with the minimum $\Psi$, it nevertheless demonstrates the effectiveness of Probabilistic Salience in detecting highly concentrated subtables.

\begin{figure}
\includegraphics[scale = 0.75,keepaspectratio = true,viewport = 2.5cm 2cm 20cm 24cm]{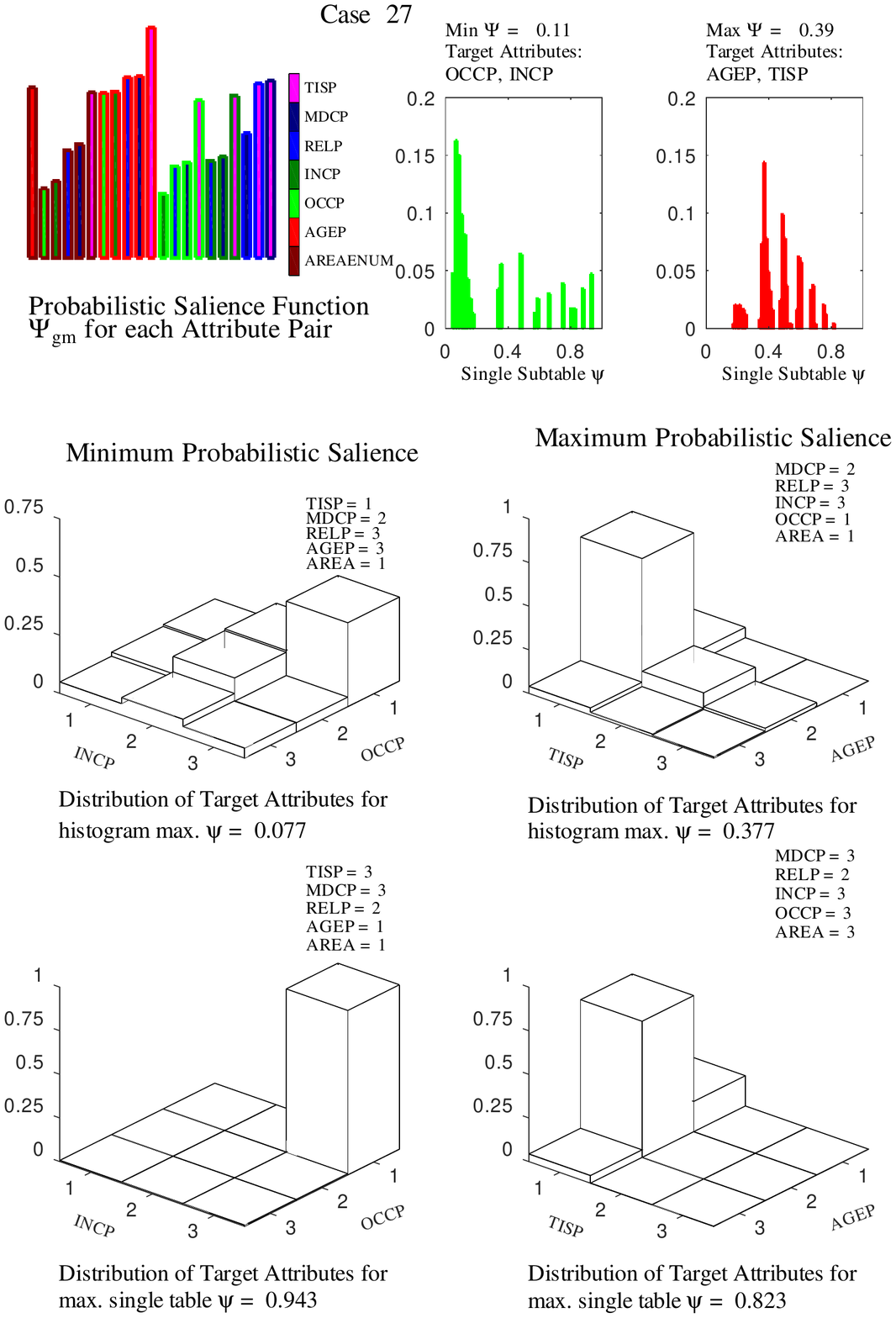}
\caption{The top left panel shows the Probabilistic Salience Function values, $\Psi$, of the geometric means of the conditional subtables for each attribute pair. At top right are histograms of the Probabilistic Salience values, $\psi$, of individual the conditional subtables for maximum and minimum $\Psi$ over the range of conditioning values for the remaining attributes. The remaining panels show three dimensional histograms of the subtables as indicated.}

\label{case27} 
\end{figure}

\section{Summary and Discussion}
\label{discussion}

What emerges from the various perspectives discussed in the Introduction is that, essentially,  social engineering is the extension into the cyber realm of the long standing crime of confidence trickery. Like the confidence trickster, the social engineer exploits psycho-social vulnerabilities in deceiving and manipulating his victims. In light of the failure by the Law to eliminate traditional confidence trickery, the expectation that security technology will defeat social engineering appears to be unreasonably optimistic. The underlying goal of social engineering, in fact, is to bypass security measures such as authentication by focusing on the human element. Indeed authentication mechanisms themselves are susceptible to exploitation \cite{Grimes}\cite{Siadati}.

In response, defence strategies to counter social engineering attacks have concentrated on strengthening the human element by providing strong security policies backed up by ongoing security awareness  and resistance training. A complementary approach is to increase the difficulty of social engineering by e.g. devising more robust SMS messages in two factor authentication \cite{Siadati} or by detecting exploitable personal information spread across social media and reducing its vulnerability. A third strategy, which is of particular interest to us, is developing technical aids to provide assistance in detecting and rebuffing an attack such as the Topic Blacklist system described in \cite{B&H}. 

We focus on those areas of social engineering which involve impersonation based on pretexting where a targeted employee, such as a call centre operator or IT support person, is persuaded that the social engineer is the intended victim in order to obtain elements of the victim's critical identifying information. This requires the social engineer to create a fictitious scenario, the pretext, based on sufficient personal information about the victim to convince the target of its veracity. Note that because the target does not know everything about the victim, the information does not have to be entirely accurate but the scenario as a whole does need to be coherent and plausible. 

The underlying issue is that the social engineer is unlikely to be able to assemble enough factual information to construct a sufficiently complete profile of the victim to provide that sense of coherence. However enough gaps can be filled in by inferring missing information from general background data that the social engineer can acquire through thorough research. We have argued in Section \ref{inference} that the form of inference is analogical taking place in an explanatory context. However the psychological evidence suggests that this rests upon a sense of statistical relevance associated with the background data meaning, roughly, that some things appear intuitively more likely than others. We argue that this sense of statistical relevance is grounded in a perhaps informal knowledge of statistical data.

Our central point is that some elements, i.e. attribute values, of a victim's profile are much more likely to have been inferred than others and these attributes need to be identified as unreliable in the sense that they do not add significantly to the profile's veracity and so should be discounted. Further, an estimate of reliability can be obtained by analysing the formal statistical data relevant to the problem based on identifying subsets of the profile which have high probability given the other elements of the profile.

We assume that the data is in the form of a multivariate contingency table (Section \ref{contingency}) so that the mathematical problem is that of analysing conditional association in the table using a log linear technique \cite{Bergsma}. In this case we use a very specific form of log linear model based on an orthogonal transformation of the logarithms of the table entries into an `interaction space' which reveals the statistical interactions between the data elements. The magnitude of the transformed table vector is the basis of our definitions of Probabilistic Salience and the Probabilistic Salience Function (Section \ref{salience}).

Estimating the reliability of profile attributes is then the problem of identifying subsets of profile attributes having values which occur frequently in the population. Consequently these values can be inferred probabilistically rendering the attributes unreliable. This becomes the problem of identifying conditional subtables of the contingency table in which the subtable contents are concentrated in a small number of cells. 

Because there is no simple topological relationship between these cells in general, conventional measures of centrality such as variance are not relevant leading us to define our Probabilistic Salience measures. This, we show, does indeed provide an indication of concentration in this sense (Section \ref{vindication}). Our computational results, while based on a limited data set, do support these theoretical results.

The Probabilistic Salience Function, then, is an indicator of the potential for exploitation of informal statistical information by social engineers in impersonating an individual. This can be applied to various subsets of attributes to detect vulnerabilities and can be used to warn potential targets such as operators in call centres that particular attributes are unreliable indicators of identity. 

The underlying theme of the literature review in Section 1.1 is that the most effective method of defending against social engineering is to maximise the support provided to the humans who face these attacks. We envisage that the probabilistic salience measures can be used, not only as a stand alone warning, but also as a component of a more comprehensive warning system. This would include other forms of social engineering detector such as the Topic Blacklist \cite{B&H} or indicators extracted from the knowledge graph of \cite{Wang2} as discussed there. Indeed our proposed probabilistic salience indicator could conceivably be used to enhance the automatic social engineering vulnerability scanner described in \cite{Edwards}.

The principal objective of the probabilistic salience indicator is to alert the target to the presence, in the information presented by the attacker, of information that could have been inferred rather than be reliable factual information. In the context of the RISP model of information processing, this should increase the target's information insufficiency level to the point where the target is driven to process the attacker's claims systematically rather than heuristically, significantly increasing the chances of detecting the flaws and inconsistencies that would indicate the claims are false.

It is conceivable that RISP concepts could be incorporated into security awareness training \cite{Saleem}, \cite{Fan}, \cite{B&H} \cite{Grimes} in order to enhance a target's response to a probabilistic salience alert by increasing the target's information sufficiency threshold \cite{Luo}. This would increase the target's information insufficiency level thereby making the target more sensitive to the probabilistic saliency alert as well as other forms of alert. Indeed because we have focussed on only one aspect of RISP, it seems possible that other aspects could also be incorporated into security awareness training suggesting that this could be a fruitful area of future research.

Finally we recognise that our Probabilistic Salience measures are a double edged sword in the sense that social engineers could use the same process to find inferable attribute values to exploit. This leads to the suggestion that contingency table data of the type used here could be `sanitised' similarly to data bases modified using differential privacy techniques. In this case we speculate in Section \ref{limiting} that the sensitivity of contingency table data to inference could be reduced by limiting the magnitude of the transformed table vector in `interaction space' thus reducing the Probabilistic Salience measures without significantly distorting the data.

\section{Acknowledgment}  \nonumber
The financial support for GS as well as the provision of meeting facilities by CSIRO Data 61 is gratefully acknowledged. We also appreciate the comments of the anonymous reviewers which resulted in a substantial improvement to the initial submission.

\end{document}